\documentclass[conference]{IEEEtran}
\IEEEoverridecommandlockouts

\usepackage{cite}
\usepackage{amsmath,amssymb,amsfonts}
\usepackage{graphicx}
\usepackage{textcomp}
\usepackage{xcolor}

\usepackage{algorithm}
\usepackage{algpseudocode}
\usepackage{mdwmath}

\usepackage[font=small,skip=4pt]{caption}
\usepackage[font=small,skip=0pt]{subcaption}
\usepackage{amsthm}
\newtheorem{theorem}{Theorem}
\newtheorem{lemma}{Lemma}

\usepackage{url}

\usepackage{pgfplots}
\usepackage{tikz}

\def\BibTeX{{\rm B\kern-.05em{\sc i\kern-.025em b}\kern-.08em
    T\kern-.1667em\lower.7ex\hbox{E}\kern-.125emX}}

\newcommand{\prot}{{\sc DEMBA}}

\begin{document}

\title{\underline{D}ecentralized \underline{E}xchange that \underline{M}itigate a \underline{B}ribery \underline{A}ttack
}


\author{Nitin Awathare \\IIT Jodhpur\\nitina@iitj.ac.in}

\maketitle

\begin{abstract}
Despite the popularity of Hashed Time-Locked Contracts (HTLCs) because of their use in wide areas of applications such as payment channels, atomic swaps, etc, their use in exchange is still questionable. This is because of its incentive incompatibility and susceptibility to bribery attacks. 

State-of-the-art solutions such as MAD-HTLC (Oakland'21) and He-HTLC (NDSS'23) address this by leveraging miners' profit-driven behaviour to mitigate such attacks. The former is the mitigation against passive miners; however, the latter works against both active and passive miners. However, they consider only two bribing scenarios where either of the parties involved in the transfer collude with the miner. 

In this paper, we expose vulnerabilities in state-of-the-art solutions by presenting a miner-collusion bribery attack with implementation and game-theoretic analysis. Additionally, we propose a stronger attack on MAD-HTLC than He-HTLC, allowing the attacker to earn profits equivalent to attacking naive HTLC.

Leveraging our insights, we propose \prot, a game-theoretically secure HTLC protocol resistant to all bribery scenarios. \prot\ employs a two-phase approach, preventing unauthorized token confiscation by third parties, such as miners. In Phase 1, parties commit to the transfer; in Phase 2, the transfer is executed without manipulation. We demonstrate \prot's efficiency in transaction cost and latency via implementations on Bitcoin and Ethereum.
\end{abstract}


\section{Introduction}
\label{sec:introduction}

Cryptocurrencies such as Bitcoin~\cite{nakamoto2008bitcoin} and Ethereum~\cite{Buterin2015ANG}, which are based on blockchain technology, allow for the secure transfer of tokens without the need for a central authority. They facilitate straightforward internal token transactions and allow for the implementation of sophisticated smart contracts for the conditional transfer of tokens. Smart contracts are created and interacted with through transactions, and blocks containing these transactions are generated by entities known as miners. 
Miners construct and publish blocks, forming the blockchain, validating transactions, and driving the system forward. The system's state is derived by sequentially processing transactions in the order they appear within the blocks, starting from the genesis block to the most recent one. 


The Hashed TimeLocked Contract (HTLC)~\cite{HTLC} is a commonly used smart contract that is compatible with both Ethereum and Bitcoin. It plays a significant role in the Lightning network~\cite{LNPaper}~\cite{1ML}, enabling secure payment routing through multiple payment channels, as well as facilitating contingent payments~\cite{Campanelli2017, Banasik2016, contigentPayment1, Fuchsbauer2019, Bursuc2019}, atomic swaps~\cite{Maurice2018, Malavolta2018, Meyden2019, Mahdi2019, Zie2019}, vault~\cite{Moser2016, McCorry2018, Bishop2019, Zamyatin2019}, and other applications. Essentially, an HTLC is defined by two parameters: a timelock $T$ and a hash lock $H$, which together ensure the conditional transfer of $v$ tokens from the payer (Bob) to the payee (Alice)~\footnote{Please note that throughout this paper, we will use "Alice" and "Payee" interchangeably, as well as "Bob" and "Payer}. Specifically, Alice can spend the $v$ tokens by submitting a transaction containing a pre-image of $H$ to the blockchain before the $T$ timeout, while Bob can spend the tokens after the $T$ time has elapsed. 

Regrettably, the Hashed TimeLocked Contract (HTLC) is susceptible to incentive manipulation attacks known as bribery attacks. This is because HTLC assumes that miners will promptly add Alice's transaction to the blockchain before the $T$ timeout. However, as profit-driven rational agents, miners may not conform to the expected behaviour when incentivized by a malicious Bob. For instance, Winzer et al.~\cite{Winzer2019} have demonstrated that Bob can bribe miners using specific smart contracts to disregard Alice's transactions until the timeout is reached. Likewise, Harris and Zohar~\cite{Harris2020} indicate that Bob can delay the confirmation of Alice's transaction by overloading the system with his own transactions. Both of these strategies enable Bob to acquire the HTLC tokens while denying Alice access to them, even if Alice has already published the preimage.

MAD-HTLC is an ingenious solution put forth by Tsabary et al.~\cite{Tsabary2021} to safeguard HTLC from bribery attacks. This proposal involves requiring the payer to provide collateral, and any misconduct by the payer will result in the miners seizing the collateral. 
However, MAD-HTLC focuses solely on passive mining strategies, whereby miners only conduct standard mining by selecting the best available transaction rather than creating better ones themselves. Due to this limitation, Wadhwa et al.~\cite{SarishtWadhwa2022} put forward an alternative approach called He-HTLC, which is demonstrably secure against both passive and active rational strategies of the miner. This includes strategic actions beyond mining, thereby offering a more comprehensive solution.

However, He-HTLC~\cite{SarishtWadhwa2022} as well as MAD-HTLC~\cite{Tsabary2021} don't consider the bribery scenario where a miner bribes another miner and collude to confiscate the HTLC token\footnote{\bf Note that a miner bribing another miner differs from bribery by other entities, such as Alice or Bob, because miners are not direct participants in the underlying exchange.}, which is prominent in the real-world as miners are rational entities in the network. We demonstrated it by proposing an attack called Miner to Miner Bribery Attack (M2MBA). 

M2MBA exploits He-HTLC when miners bribe another miner and collude to confiscate the tokens involved in Bob's collateral. In He-HTLC, Exchange tokens~\footnote{The tokens that Bob sends to Alice} are transferred to Alice, and collateral is transferred to Bob if Alice reveals his preimage before time $T$. On the other hand, after time $T$, Bob can redeem both exchange tokens and collateral by revealing his preimage. However, if a miner has both the preimages, then they can confiscate Bob's collateral and the exchanged tokens are burned. In M2MBA, miners would wait for Bob to reveal his preimage, say $pre_B$, even if Alice reveals his preimage, say $pre_A$ and hence confiscate Bob's collateral once Bob reveals $pre_B$. Note that a miner can execute this attack alone However, by bribing and colluding with other miners-as done in M2MBA-the miner can achieve a more stable income (demonstrated in Appendix~\ref{apx:soloVsPool}).



Additionally, we propose an attack on MAD-HTLC called the Bribery-Based Block Broadcasting Attack (B3A), which is more severe than the SDRBA and HyDRA attacks proposed by He-HTLC. In B3A, Bob gains more tokens than he would through SDRBA or HyDRA, providing even stronger incentives to launch the attack. In fact, our analysis shows that Bob can extract nearly the same amount of profit as he would by attacking the naive HTLC.

To address these issues, we propose \prot, which ensures that miners cannot confiscate tokens as long as neither Alice nor Bob deviates from the protocol. Simply removing the miner's confiscation power would reintroduce vulnerabilities of naive HTLCs. Instead, \prot\ mitigates this by requiring both parties to post collateral, shifting accountability to the participants themselves. This design is based on the key insight that the core vulnerability in MAD-HTLC and He-HTLC stems from granting enforcement power to an uninvolved third party—the miner.

\prot\ operates in two phases. In the first phase, each party redeems their respective collateral by revealing their preimage, thereby committing to the exchange. Once this commitment is made--by broadcasting a transaction that reveals a specific preimage--they can no longer back out. Specifically, Alice can commit in one of three ways, depending on when and which preimages she reveals (\( pre_A \) and/or \( pre_A' \)). First, revealing \( pre_A \) transfers the exchange amount \( v^{\text{dep}} \) to Alice. Second, revealing only \( pre_A' \) returns \( v^{\text{dep}} \) to Bob. Lastly, revealing both \( pre_A \) and \( pre_A' \) burns \( v^{\text{dep}} \). Alice makes this final commitment only if she previously revealed \( pre_A \), as a means to punish Bob for bribing miners to censor her commitment. This mechanism deters bribery by ensuring that any attempt to suppress a valid transaction leads to a loss for Bob.

On the other hand, Bob has one way to commit--by revealing his preimage \( pre_B \). Once both Alice and Bob have made their respective commitments, the protocol proceeds to the second phase, where the exchange is executed. Our game-theoretic analysis in Section~\ref{sec:DEMBAAnalysis} shows that any deliberate false commitment results in a loss for the deviating party.



Lastly, We implemented and evaluated \prot\ on Bitcoin and Ethereum. On Bitcoin, \prot\ reduces costs by 42\% vs. MAD-HTLC but incurs 17\% more than He-HTLC for Alice's token claim, while lowering Bob's refund cost for both his collateral and exchange tokens by 42\% vs. He-HTLC and 150\% vs. MAD-HTLC. Similar trends hold on Ethereum, though with smaller margins: \prot\ saves 38\% vs. MAD-HTLC and costs 7\% more than He-HTLC for Alice, while reducing Bob's refund cost by 23\% and 71\% vs. He-HTLC and MAD-HTLC, respectively.

Furthermore, We compared \prot\ with He-HTLC and MAD-HTLC in terms of time to complete (TTC), defined as the duration from transaction initiation to completion. \prot\ matches or exceeds their performance while incurring lower costs.


In summary we made following contributions:
\begin{itemize}
    \item We demonstrate two attacks on MAD-HTLC and He-HTLC, showing that Bob can still exploit MAD-HTLC to earn the same tokens as in an attack on naive HTLC.
    
    \item We propose \prot, the first HTLC scheme addressing all bribery scenarios—Alice bribing miners, Bob bribing miners, and miner-to-miner bribery—unexplored in prior work. Our claims are supported through game-theoretic analysis.

    \item We implemented and evaluated \prot\ on Bitcoin and Ethereum, demonstrating its superiority over state-of-the-art protocols in most cases.
\end{itemize}





 The rest of the paper is organized as follows: \S\ref{sec:background} discusses existing work that aims to enhance naive HTLC to resist bribery attacks.
 The motivation is explained in \S\ref{sec:attacks} through the demonstration of attacks on existing work. \S\ref{sec:overview} gives an overview of \prot. The game-theoretic analysis of the M2MBA attack and \prot\ are presented in \S\ref{sec:M2MBAttackAnalysis} and \S\ref{sec:DEMBAAnalysis}, respectively. \prot’s implementation, evaluation, and results are discussed in \S\ref{sec:DEMBAEvaluation}. \S\ref{sec:relatedwork} reviews related work. Finally, the paper concludes with a discussion and future directions in \S\ref{sec:futurework} (In the Appendix)





\section{Background}
\label{sec:background}
To facilitate comprehension and provide context, we will begin by revisiting Hash Time-Locked Contracts (HTLC), outlining the issues associated with the current HTLC protocol, and examining prior attempts to address them before delving into our findings.



\subsection{Naive HTLC and Bribery attack on it}
A naive Hash Time-Locked Contract (HTLC), $C^{dep}_{B->A}$,  involves Bob paying Alice $v^{dep}$ tokens is defined by the tuple ($pk_A$, $pk_B$, $v^{dep}$, $pre_A$, $T$). The $C^{dep}_{B->A}$ indicates that Bob's deposit of $v^{dep}$ tokens which either transferred to Alice when Alice broadcast
signed transaction $tx^{dep}_A$ that reveals $pre_A$ or Bob can redeem by broadcasting a signed transaction $tx^{dep}_B$ after the timeout period $T$. The later one serves as a refund mechanism if Alice remains inactive till time $T$. It is important to note that the preceding description omits implementation- and application-specific details, enabling us to concentrate on the fundamental concerns.



Winzer et al.~\cite{Winzer2019} and Tsabary et
al.~\cite{Tsabary2021} demonstrated attacks on the naive HTLC by illustrating a scenario where Bob bribes miners to withhold $tx^{dep}_A$ until time $T$. Once $T$ is reached, Bob can claim $v^{dep}$ using $tx^{dep}_B$ and offering a transaction fee of $f_B^{dep}$. The profit gained from this attack can be calculated as follows. Let $t_{pub}$ be the time when Alice broadcast $tx^{dep}_A$ offering a transaction fee of $f_A^{dep}$ and $k$ be the number of block generated between time $t_{pub}$ and $T$. It is further assumed that Bob offers a bribe $br > f_A^{dep}$ to the miner, per block, for not including $tx^{dep}_A$. Consequently, Bob gains $v^{dep} - [(k+1)*br+f_B^{dep}+f_B^{C_{Bob}}]$ after bribing the miner, as shown in Figure~\ref{fig:HTLC_Bribery}, where $k = 4$.
We have given the sample smart contract, say $C_{Bob}$, to perform the above mentioned attack in Algorithm~\ref{algo:Bribery-Contract}. In the forthcoming sections, we will demonstrate our proposed attacks on the existing remedies for the pitfalls of the naive HTLC using the modified version of the same contract. 
Given the contract $C_{Bob}$, the attack proceed as follows:
\begin{itemize}
    \item Bob deploy a smart contract $C_{Bob}$, described  in Algorithm~\ref{algo:Bribery-Contract}, and initialize it with $v^{dep}$. This costs Bob the fee of $f_B^{C_{Bob}}$.
    \item Miners mining the block between time $t_{pub}$ and $T$ call {\em RequestBribe} to reserve their potential bribe for censoring $tx_A^{dep}$.
    
    \item After time $T$ Bob broadcast $tx_B^{dep}$ to claim $v^{dep}$ (refund), which miners include in their block. This costs Bob the transaction fee of $f_B^{dep}$.
    \item After $tx_B^{dep}$ is included in the block, any miner can call $claimBribe$ with a $pre_A$, which when executed send the reserved bribe to all the miners who has censored the $tx_A^{dep}$. The miner who include the transaction calling $claimBribe$ will receive the fee of $br$, which is the extra $br$ in the equation $v^{dep} - [(k+1)*br+f_B^{dep}+f_B^{C_{Bob}}]$. 
\end{itemize}

\begin{figure}[t!]
    \centering
    \includegraphics[height=2.65cm, width=0.98\linewidth]{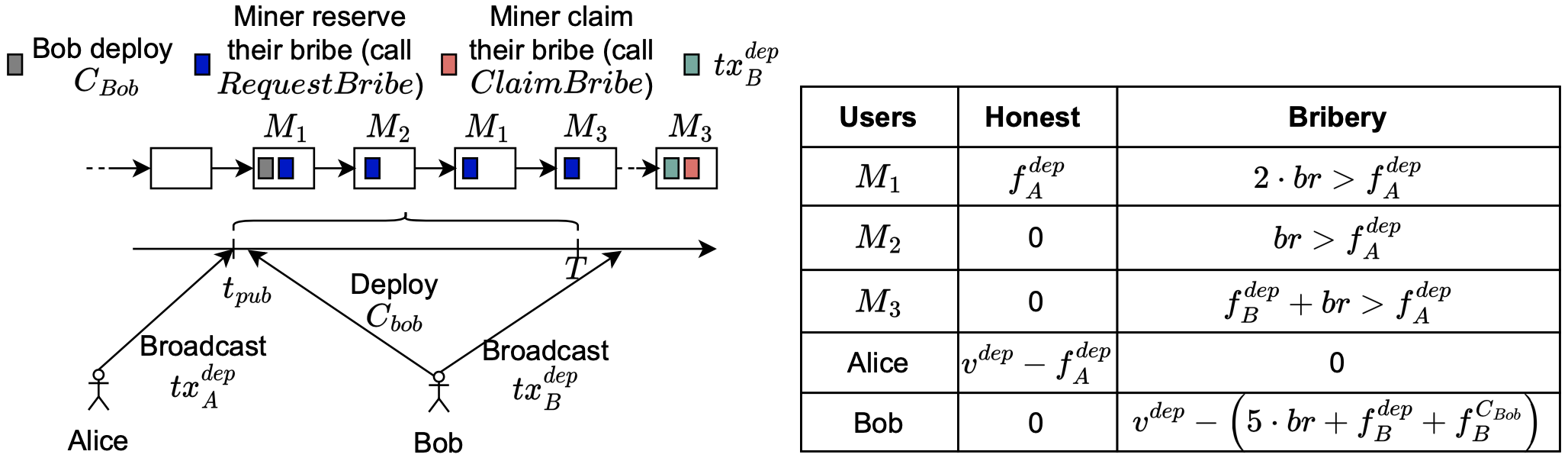}
    \caption{Bobs' attempt to make an bribery attack on the naive HTLC. Bob deploy contract committing a bribe $br$ to the miners for censoring $tx^{dep}_A$, which is ultimately send to miners at time $t > T$ }
    \label{fig:HTLC_Bribery}
\end{figure}

\subsection{Existing work overcoming the bribery attack on naive HTLC (MAD-HTLC and He-HTLC)} 

To address the issue of HTLC bribery, Tsabary et al.~\cite{Tsabary2021} proposed MAD-HTLC, which features two key modifications. Firstly, it adds a second hash lock with preimage $pre_B$, along with an extra redemption path that allows miners to redeem $v^{dep}$ if they possess both $pre_A$ and $pre_B$. Secondly, MAD-HTLC introduces a collateral contract called {\em MH-Col}, initiated by Bob, that contains collateral tokens $v^{col}$ that miners can also seize if they know ($pre_A, pre_B$). These changes are intended to discourage Bob from revealing $pre_B$ unless he genuinely requires a refund, particularly when $pre_A$ has not yet been released.

\begin{figure*}[t!]
    \centering
    \begin{subfigure}{0.29\linewidth}
    \centering
    \includegraphics[height=3.05cm, width=0.98\linewidth]{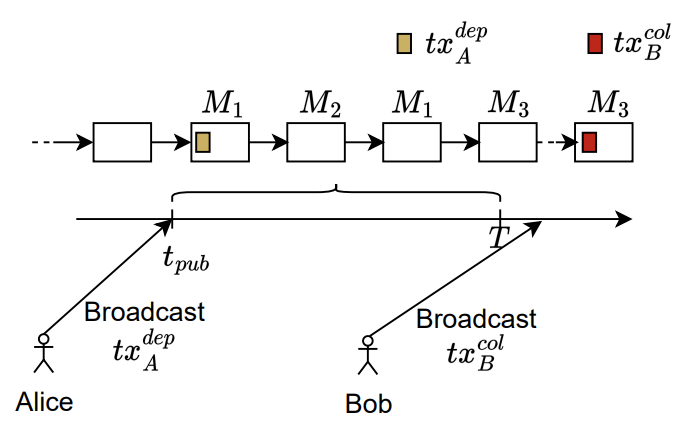}
    \caption{}
    \label{fig:MD_HTLC_Honest}
    \end{subfigure}
    \begin{subfigure}{0.31\linewidth}
    \centering
    \includegraphics[height=3.05cm, width=0.98\linewidth]{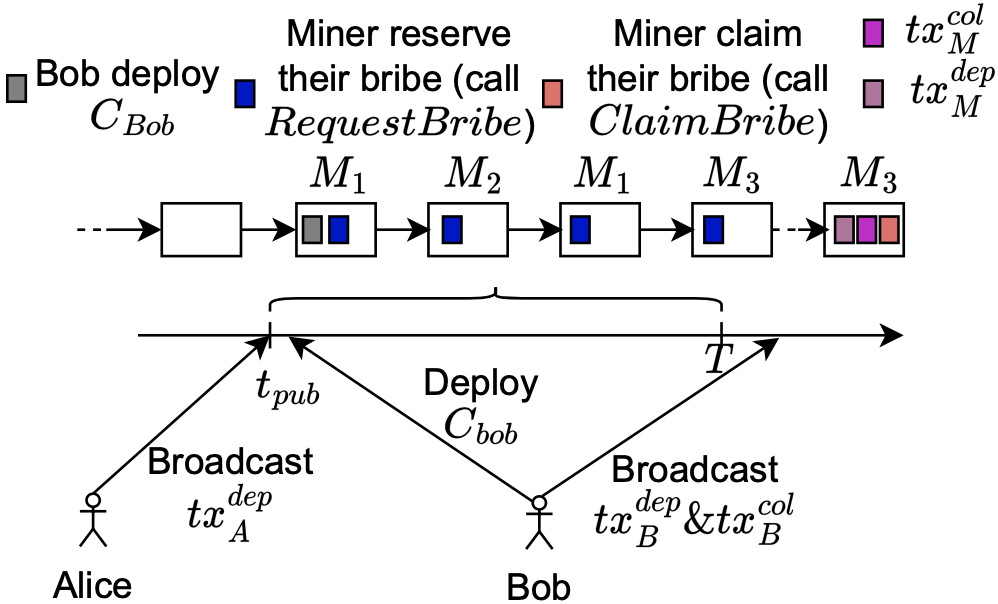}
    \caption{}
    \label{fig:MD_HTLC_Attack}
    \end{subfigure}
    \begin{subfigure}{0.31\linewidth}
    \centering
    \includegraphics[height=3.05cm, width=0.98\linewidth]{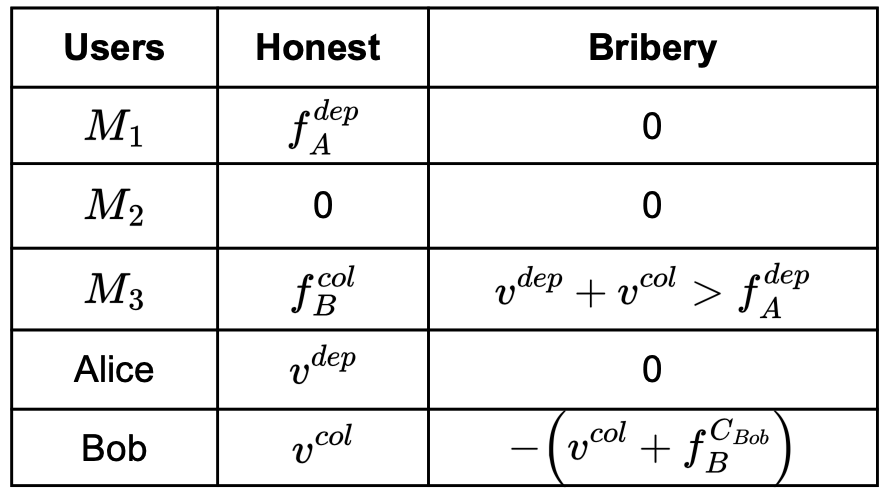}
    \caption{}
    \label{fig:MD_HTLC_table}
    \end{subfigure}
    \caption{(a)Honest execution of MAD-HTLC, (b)Bob perform the bribery attack mentioned above on the MAD-HTLC, (c)Utility gained by each party during the honest and malicious execution of MAD-HTLC}
    \label{MAD_HTLC_execution}
\end{figure*}


In MAD-HTLC, the payment of $v^{dep}$ tokens made by Bob is stored in a contract called MH-Dep, which outlines three redemption paths: 1) {\em dep-A: } where Alice can redeem $v^{dep}$ by broadcasting a signed transaction $tx_A^{dep}$ revealing $pre_A$, 2) {\em dep-B: } where Bob can redeem $v^{dep}$ by broadcasting a signed transaction $tx_B^{dep}$ after time $T$, and 3) {\em dep-M: } where anyone (e.g., a miner) can redeem $v^{dep}$ by broadcasting the transaction $tx_M^{dep}$ revealing both $pre_A$ and $pre_B$. Bob's collateral $v^{col}$ is stored in contract MH-Col, which can be redeemed in two ways: 1) {\em col-B: } where Bob can redeem by broadcasting a signed transaction $tx_B^{col}$ after $T$, and 2) {\em col-M: } where anyone (miner) can redeem by broadcasting the transaction $tx_M^{col}$ revealing both $pre_A$ and $pre_B$. 

As depicted in Figure~\ref{fig:MD_HTLC_Attack}, rational miners at time $t>T$ will not allow Bob to spend $v^{dep}$ using $tx^{dep}_B$ if they know $pre_A$. Instead, they will confiscate both $v^{dep}$ and $v^{col}$ through {\em dep-M} and {\em col-M}, respectively. Consequently, Bob loses not only $v^{col}$ and $v^{dep}$ but also $f^{C_{Bob}}_B$, which the fee charged by the miners to deploy contract $C_{Bob}$. On the other hands, the miner who confiscates ($M_3$ in our example) earns $v^{dep}$ and $v^{col}$, as illustrated in Figure~\ref{fig:MD_HTLC_table}. 



However, MAD-HTLC does not address the possibility of miners actively seeking out other users in the system and bribing them to obtain a higher utility, called as reverse bribery. This is because the contract allows the miners and Bob to redeem $v^{col} + v^{dep}$.
A successful attack could split this total gain such that Bob and miners are individually better off than following the protocol. ~\cite{SarishtWadhwa2022} presents two such attacks that attempt reverse bribery: 1) Success-dependent reverse bribery attack (SDRBA), and 2) Hybrid delay-reverse bribery attack (HyDRA). The details of SDRBA and HyDRA is presented in Appendix~\ref{apx:SDRBA}.

The SDRBA and HyDRA rely on two main factors. First, miners collude with either Bob or Alice to earn a high profit, $v^{dep}+v^{col}$, through {\em dep-M + col-M}. For instance, a miner can collude with Alice to earn $v^{col} - \epsilon$ or with Bob to earn $v^{dep} - \epsilon$, which they could have earned $f^{dep}_A + f^{col}_B$ by acting honestly. Second, the redemption of $v^{dep}$ and $v^{col}$ can occur atomically in a single transaction or in two separate transactions within a block. Considering these observations ~\cite{SarishtWadhwa2022} introduces the He-HTLC to overcome the the above-mentioned attacks (SDRBA and HyDRA). He-HTLC achieve this by proposing two contracts 1) {\em He-Dep}, and 2) {\em He-Col}, like MAD-HTLC. The payment of $v^{dep}$ and $v^{col}$ is stored by Bob in {\em He-Dep}, which can be redeemed in two ways: 1) {\em dep-A: } Transfers $v^{dep}$ to Alice and $v^{col}$ to Bob, when Alice reveals $pre_A$. 2) {\em dep-B: } Deposit $v^{dep} + v^{col}$ to {\em He-Col} after time $T$, if Bob reveals $pre_B$. Similarly, {\em He-Col} can be redeemed in two ways: 1) {\em Col-B: } Transfer $v^{dep} + v^{col}$ to Bob after time $T+l$, and 2) {\em Col-M: } Transfer $v^{col}$ to the miner producing both $pre_A$ and $pre_B$ while $v^{dep}$ are burned.    

With the above amendments, miners high earning is prevented by burning $v^{dep}$ and hence allowing them to confiscate only $v^{col}$ instead of $v^{dep} + v^{col}$. On the other hands, Bob's high earning is prevented by requiring him to disclose the secret $pre_B$ (through the {\em dep-B} path) for $l$ blocks before he can receive $v^{dep} + v^{col}$ via the {\em col-B} path. This separation provides two key benefits. Firstly, if $pre_A$ is unavailable, Bob will receive $v^{dep} + v^{col}$ after a delay of $l$ blocks. Secondly, if $pre_A$ is known, miners can choose between two options: either they can confiscate $v^{col}$ through {\em col-M}, or they can allow Bob to receive $v^{dep} + v^{col}$ and potentially earn more than $v^{col}$ through a bribe from Bob. The authors prove that this bribe is avoidable by selecting an appropriate value of $l$.

\section{Motivation (Attacks on the existing work)}

%
\label{sec:attacks}
We present two attacks: the Bribery-based Block Broadcasting Attack (B3A) and the Miner-to-Miner Bribery Attack (M2MBA). B3A targets MAD-HTLC, allowing Bob to extract nearly the same collateral as in a naive HTLC attack. In contrast, M2MBA targets He-HTLC, showing that miners profit more by attacking, regardless of Alice’s or Bob’s honesty.

\subsection{ Bribery based Block Broadcasting Attack (B3A)} The SDRBA and HyDRA protocols rely on the assumption of a fair exchange between the miner and Bob, wherein the miner only pays Bob a bribe if they can redeem $v^{dep}$ and/or $v^{col}$. To implement this, ~\cite{SarishtWadhwa2022} proposed the following approach: Bob creates a transaction $tx^{dep}_M$ that redeems MH-Dep for $M_i$, sends $h=H(tx^{dep}_M)$ to $M_i$, and provides a proof $\pi$ to demonstrate its correctness. The miner, $M_i$, then verifies the proof and mines a partial block $B$. This block includes: 1) The hash of $tx^{dep}_M$, which is noteworthy as mining is typically performed on a block header that contains a compact representation of the transaction (e.g., a Merkle root that hashes all transactions), and 2) A bribing transaction that pays Bob $br$ tokens from the coinbase, along with any other transactions from the mempool. Finally, the miner sends $B$ to Bob, who checks that the block contains the intended transactions, completes $B$ by adding $tx^{dep}_M$, and broadcasts the finalized block.

By broadcasting the completed block, Bob gains leverage to perform the bribery attack and extract additional collateral—matching what he would gain from bribing a naive HTLC. He finalizes and broadcasts the block only if it includes either of the following, along with unrelated transactions from the pool:

\begin{itemize}
    \item case -1:  a) The bribing coinbase transaction that pays Bob a bribe of $v^{dep}-br$, b) Hash of $tx^{dep}_M$ which redeems $v^{dep}$ for miner, c) $tx_B^{col}$, which redeem $v^{col}$ for Bob at cost of paying a transaction fee of $f^{col}_B$ to the miner, d) Transaction calling {\em claimBribe} of $C_{Bob}$. 
    
    \item case -2:  a) The bribing coinbase transaction that pays Bob a bribe of $v^{dep}+v^{col} - 2br$, b) Hash of $tx^{dep}_M$ and $tx^{col}_M$ which redeems $v^{dep}$ and $v^{col}$ for miner, respectively, c) Transaction calling {\em claimBribe} of $C_{Bob}$.
\end{itemize}

Note that the above-mentioned block may appear after the timeout \( T \). As a result of the attack, Bob must pay a bribe of \( k \cdot br \) to each miner who mines a block between \( t_{\text{pub}} \) and \( T \), and an additional \( br \) to the miner who includes the relevant block after \( T \). Moreover, the miner should receive \( br \) for including \( tx_M^{\text{dep}} \), and either \( f_B^{\text{col}} \) for including \( tx_B^{\text{col}} \) (Case 1), or \( br \) for including \( tx_M^{\text{col}} \), compensating for the fees they forgo by not including unrelated transactions or \( tx_A^{\text{dep}} \) and/or \( tx_B^{\text{col}} \). 
Consequently, Bob's net gain is calculated as $v^{dep} - [(k+2)*br+ (br $ \text{or} $f^{col}_B)+ f_B^{C_{Bob}}]$ where $k$ is the number of blocks mined in time $T - t_{pub}$. Assuming that $f^{col}_A = f^{dep}_B$, Bob can carry out the bribery attack with an additional cost of $br$ compared to the attack on the naive HTLC. 

It is important to note that B3A differs from the Success Independent Reverse Bribery Attack described in He-HTLC, which assumes that the miner refrains from sharing $pre_B$ with any miner other than the briber. In contrast, B3A operates without making this assumption.\\

\subsection{ Miner to Miner Bribery Attack (M2MBA) }
\label{sec:M2MBABrief}
The above solutions proposed in He-HTLC guard when either party (Alice or Bob) colludes with miner. More precisely they guards against following two situations: 
\begin{itemize}
    \item Alice has colluded with a miner to receive greater incentives than she would by behaving honestly. This has been accomplished through the use of $v^{col}<v^{dep}$.
    
    \item Bob has conspired with a miner to receive greater incentives than he would through honest behavior. This is accomplished by permitting the miner to seize only $v^{col}$ and incinerate $v^{dep}$ when both $pre_A$ and $pre_B$ are available.

\end{itemize}


However, although Alice and Bob may act honestly, the possibility of bribery scenarios among miners where one miner bribes others and shares the high earnings gained from the attack is often overlooked in the proposed solutions. Here we will demonstrate that an miner can collaborate with other miners to earn more than they would have earned through honest behavior. If the miner chooses not to disclose $pre_A$ on the chain, other miners may be willing to do the same, even if Alice discloses $pre_A$, as long as there is a potential for greater profit. The detailed attack steps are outlined as follows:

\begin{itemize}
    \item Alice will reveal $pre_A$ to the miner. However, the miner won't include the respective transaction ($tx_A^{dep}$) in the block, hence hiding $pre_A$. It also promises bribes to the other miners for not including $tx^{dep}_A$ in the block.
    \item Miner waits until time $T$, after which Bob will broadcast $tx^{dep}_B$ that reveals $pre_B$ \footnote{Note that in this scenario, Alice and Bob is honest and adheres to the protocol. Since the malicious miner does not reveal any information about $pre_A$ on-chain, Bob will proceed as specified and honestly release $pre_B$ according to the protocol.}. 
    \item Miner will follow the path {\em col-M} as they know $pre_A$ and $pre_B$ to confiscate the $v^{col}$ and send the promised bribe to the other miners who censored $tx^{dep}_A$ between $t_{pub}$ and $T$. We called the miner who confiscate $v^{col}$ a {\em bribing miner}.
\end{itemize}

By behaving honestly each miner can earn $f^{dep}_A$, or $f^{col}_B$, as shown in Figure~\ref{fig:He_HTLC_Honest_A} and~\ref{fig:He_HTLC_Honest_B} respectively. The later one is the case when Alice doesn't send $tx^{dep}_A$ revealing $pre_A$.
However, performing the above-mentioned attack the bribing miner, say $M_i$, will be able to gain $v^{col} - (k-k_{M_i})*br$, where $br > f^{dep}_A$ and $k$ and $k_{M_i}$ is the total number of blocks mined and number of blocks mined by $M_i$ in time $T-t_{pub}$, respectively.
In the example showed in the Figure~\ref{fig:He_HTLC_Honest_M2MAttack}, where $k = 4$, and $k_{M_3} = 1$, $M_3$ is the bribing miner, who confiscate $v^{col}$. $M_3$ also promise the bribe to other miners for censoring $tx^{dep}_A$ by locking $v^{col}$ in contract, say $C_{M2M}$, details of which is explained further in section~\ref{sec:M2MBA}. 
Given this, $M_3$ gained $v^{col} - 3*br > f^{col}_B$, which is depicted in Figure~\ref{fig:He_HTLC_Honest_Attack_comparison_Table}. While other miners, say $M_j$, will gain $br*k_{M_j}$, where $k_{M_j}$ is the number of blocks mined by $M_j$ during time $t_{pub}$ to $T$. 

On the other hands, if $v^{col} - 3*br < f^{col}_B$, then $M_3$ (bribing miner) would prefer to include $tx^{dep}_B$ and earn $f^{col}_B$. However, because we assumed that $v^{col} > f^{dep}_A$ and $v^{col} > f^{col}_B$, it would rarely be the case. This is a legitimate assumption because the transaction fee is much lower than the transaction value.  


Here we have distributed the bribe, where the bribing miner, $M_i$, would get $v^{col} - (k-k_{M_i})*br$ and other miners would receive $br*k_{M_j}$. Another way to distribute bribes is equally distributing $v^{col}$ among all involved in the attack, i.e. bribe miners and miners responsible for censoring $tx^{dep}_A$. With such distribution each miner would earns $v^{col}*k_{M_i}/k$. For ease of explanation in our analysis in section~\ref{sec:M2MBA}, we used this (equal) bribe distribution. To overcome these attacks, we propose \prot, which we discuss next.

\begin{figure*}[t!]
    \centering
    \begin{subfigure}{0.29\linewidth}
    \centering
    \includegraphics[height=3.25cm, width=0.98\linewidth]{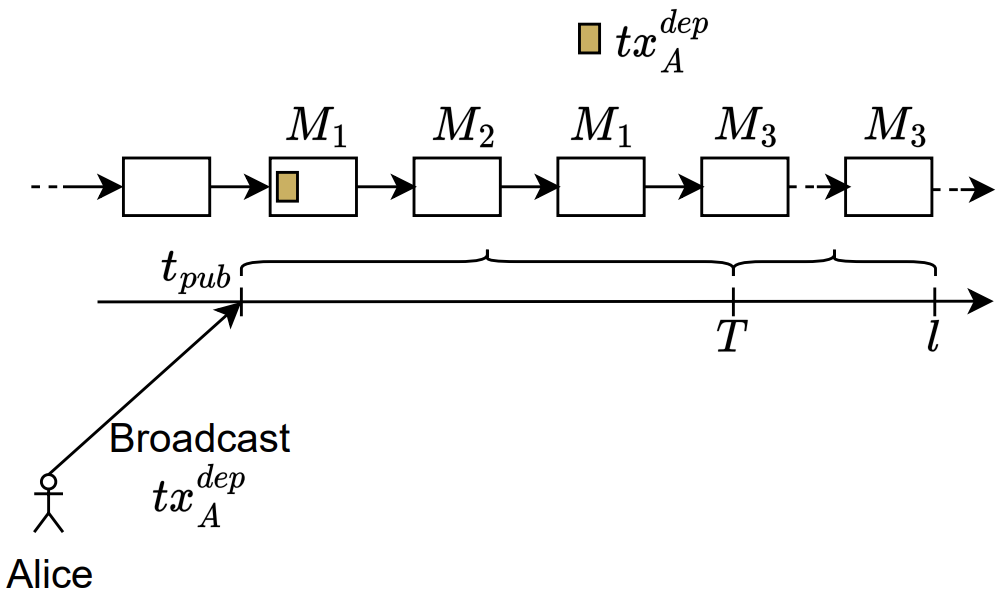}
    \caption{}
    \label{fig:He_HTLC_Honest_A}
    \end{subfigure}
    \begin{subfigure}{0.31\linewidth}
    \centering
    \includegraphics[height=3.25cm, width=0.98\linewidth]{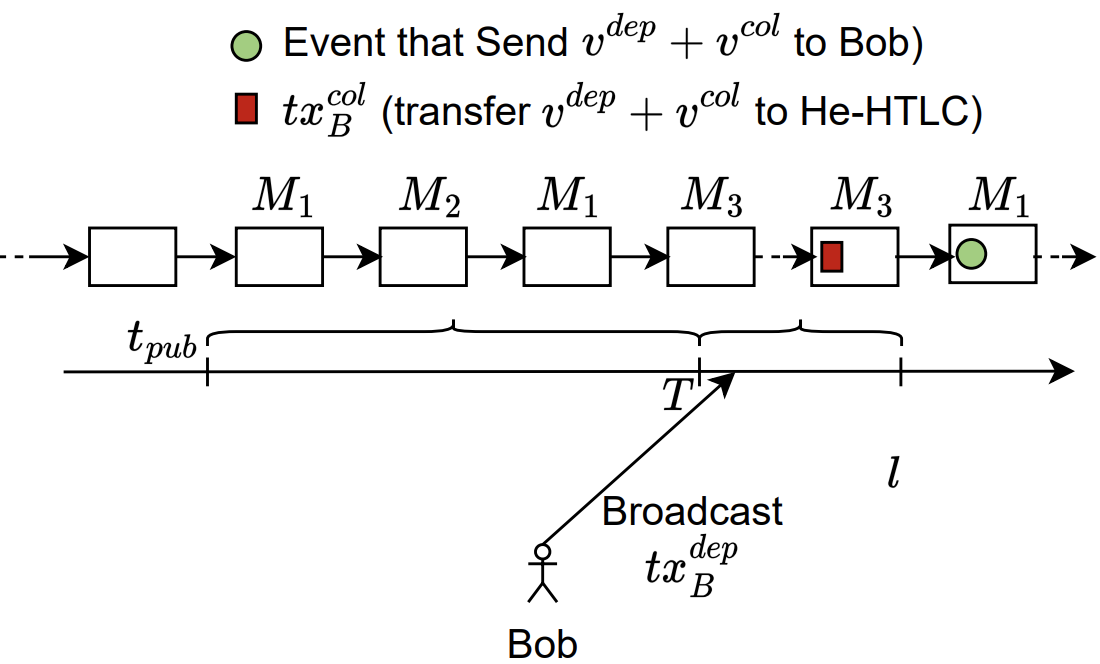}
    \caption{}
    \label{fig:He_HTLC_Honest_B}
    \end{subfigure}
    \begin{subfigure}{0.31\linewidth}
    \centering
    \includegraphics[height=3.25cm, width=0.98\linewidth]{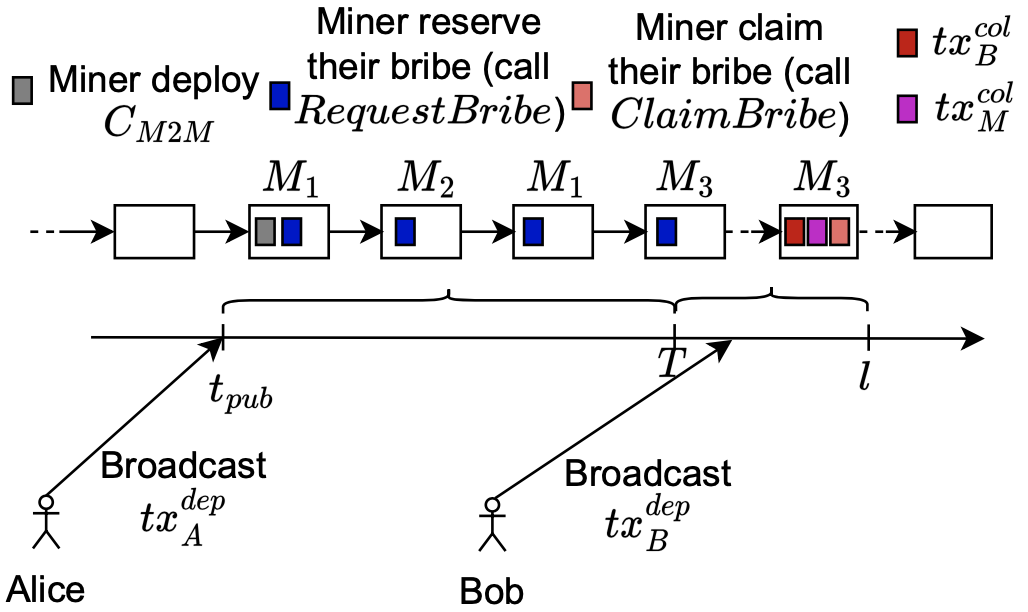}
    \caption{}
    \label{fig:He_HTLC_M2MAttack}
    \end{subfigure}    
    \caption{(a) Honest execution of He-HTLC when Alice release the $pre_A$ before time $T$, (b) Honest execution of He-HTLC when Bob release the $pre_B$ post time $T$, and (c) Miner to miner bribery attack (M2MBA) on He-HTLC.} 
    \label{fig:He_HTLC_Honest_M2MAttack}
\end{figure*}

\begin{figure}
\centering
\includegraphics[height=3.15cm, width=0.81\linewidth]{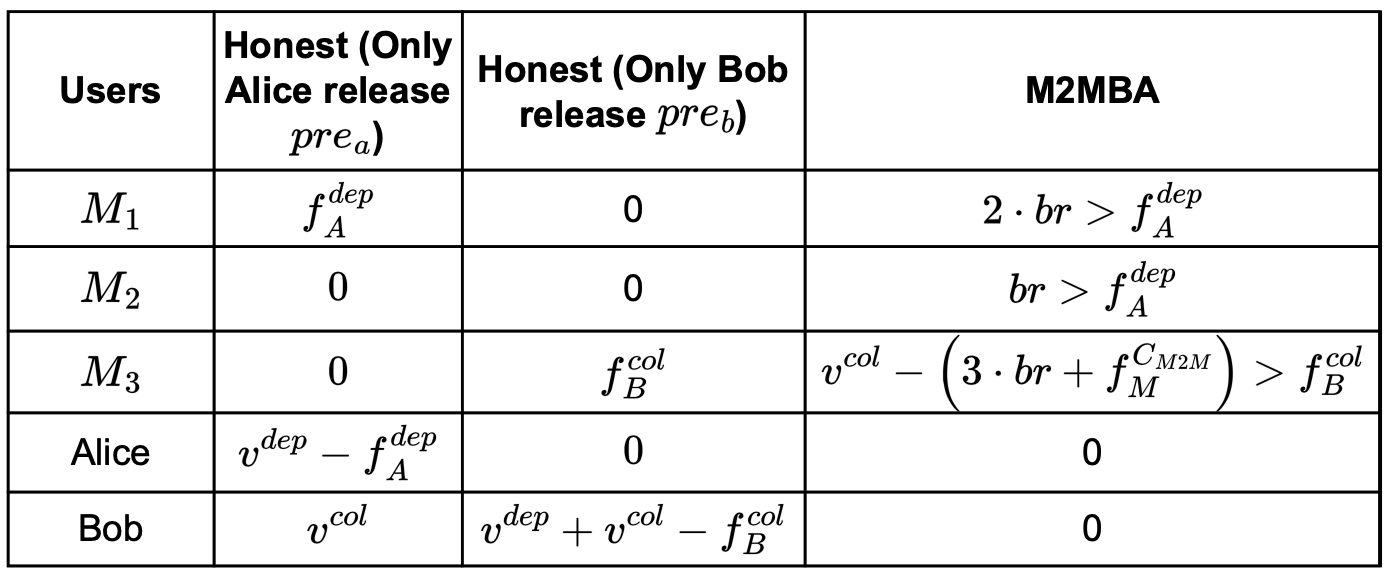}
\caption{Utility gain by each involved party during honest execution of He-HTLC and during M2MBA attack.}
\label{fig:He_HTLC_Honest_Attack_comparison_Table}
\end{figure}

\section{\prot\ Overview}
\label{sec:overview}
We begin this section by asking the question "Can we close the doors for attack using \prot ?". We provide a positive response to the preceding question by introducing a novel HTLC protocol, referred to as \prot, which is designed to be incentive-compatible for actively rational miners, taking into account all possible bribing tactics. The primary concern with the previous protocols (\cite{Tsabary2021} and \cite{SarishtWadhwa2022}) is, they allow the third entity (in our case miner), which is not a part of exchange, to confiscate the tokens ($v^{dep}$ and/or $v^{col}$). Which give them a upper hand and hence enough power to perform the attack that breaks the exchange (transfer).

We propose the solution, which revoke the power of token confiscation from the miner, in case the parties deviate from the protocol. The intuition is both the parties, Alice and Bob, involved in the exchange deposit their token in the collateral contract in addition to the contract that performs exchange. Let $C^{dep}$ be the contract that the Bob deploy to perform the exchange and deposit tokens $v^{dep}$ in it. Furthermore, let $C^{col}_A$ and the $C^{col}_B$ be the collateral contract for Alice and Bob where they deposit tokens $v^{col}_A$ and $v^{col}_B$, respectively. In addition to that, $C^{col}_A$ and $C^{col}_B$ will output the respective values that Alice and Bob used to redeem $v^{col}_A$ and $v^{col}_B$, respectively. Unlike the existing solution, $v^{dep}$ can't be redeemed manually, meaning by sending transaction that discloses the $pre_A$ and/or $pre_B$. Instead $C^{dep}$ will read the secrete from $C^{col}_A$ and $C^{col}_B$ to make a decision on whom to transfer $v^{dep}$, either Alice or Bob. 
Given this, each contract can be redeemed/unlocked in the following ways:
\begin{itemize}
    \item $C^{dep}$: 1) dep-A ($tx^{dep}_{pre_A}$): Transfer $v^{dep}$ to Alice, if $C^{col}_A$ output only $pre_A$ and $C^{col}_B$ output $pre_B$ 
    2) dep-B ($tx^{dep}_{pre_A'}$): 
    Transfer $v^{dep}$ to Bob after time $T$, if $C^{col}_A$ output only $pre_{A}^{'}$ and $C^{col}_B$ output $pre_B$. 
    3) dep-Burn ($tx^{dep}_{pre_{AA}'}$): Burn $v^{dep}$ after time $T$, if $C^{col}_A$ output both $pre_A$ and $pre_{A}^{'}$, irrespective of whether $C^{col}_B$ output $pre_B$ before time $T$.

    \item $C^{col}_A$: Alice has the option to redeem using $pre_A$ and/or $pre_{A}^{'}$. If she redeems before time $T$, she must use $pre_A$ ($tx^{col}_{pre_{A}}$). However, if she has already disclosed the transaction revealing $pre_A$, which was not recorded on the blockchain due to bribery till time $T$, she can choose to redeem using both $pre_A$ and $pre_{A}^{'}$ ($tx^{col}_{pre_{AA}'}$), else she redeem with just $pre_{A}^{'}$ after time $T$ ($tx^{col}_{pre_{A}'}$). 
    Note that, when $C^{col}_A$ is redeemed with $tx^{col}_{pre_{AA}'}$, Alice will earn $v^{col}_A - v^{ded}$.

    \item $C^{col}_B$: Bob can unlock it at anytime using $pre_B$ ($tx^{col}_{pre_{B}}$). However, After time $T$, $v^{col}_B$ is deducted by some coins, say $v^{ded}$.
    
\end{itemize}




In addition to the redemption mechanisms discussed above, \prot\ distinguishes between the transaction fee paid by the sender and the fee actually earned by miners for including a transaction in a block. Let $f^{col}_{pre_{A}}$, $f^{col}_{pre_{A}'}$, $f^{col}_{pre_{AA}'}$, and $f^{col}_{pre_{B}}$ denote the transaction fees paid by Alice and Bob for transactions $tx^{col}_{pre_{A}}$, $tx^{col}_{pre_{A}'}$, $tx^{col}_{pre_{AA}'}$, and $tx^{col}_{pre_{B}}$, respectively. Similarly, let $fm^{col}_{pre_{A}}$, $fm^{col}_{pre_{A}'}$, $fm^{col}_{pre_{AA}'}$, and $fm^{col}_{pre_{B}}$ represent the portion of those fees actually earned by miners for including the respective transactions in a block.

To incentivize honest behavior, \prot\ enforces the following conditions:
\begin{equation}
    f^{col}_{pre_{A}} < f^{col}_{pre_{A}'} < f^{col}_{pre_{AA}'}
    \label{eq:fee}
\end{equation}
\begin{equation}
    fm^{col}_{pre_{A}} > fm^{col}_{pre_{A}'} > fm^{col}_{pre_{AA}'}
    \label{eq:minerFee}
\end{equation}

For $tx^{col}_{pre_{B}}$, miners earn a higher fee if they include it before time $T$; after $T$, their earnings decrease. However, Bob pays a higher fee if the transaction is included after $T$. The increasing gap between fees paid and fees earned by miners is burned, discouraging protocol deviations. For the honest behavior the fee paid by the party and the fee earn by the miner is same, i.e. difference is zero. For example, $f^{col}_{pre_{A}} = fm^{col}_{pre_{A}}$, and $f^{col}_{pre_{B}} = fm^{col}_{pre_{B}}$, if $tx^{col}_B$ is included in the block before $T$. Such fee structure can be easily implemented using condition-based UTXOs in UTXO-based blockchains. In smart contract-based blockchains, it becomes even simpler by incorporating the condition directly within the relevant contract methods.






Given this knowledge, we can conclude that when all parties act honestly, Alice redeems $C^{col}_A$ using $pre_A$, while Bob redeems $C^{col}_B$ using $pre_B$. This provides sufficient information to unlock $C^{dep}$ and transfer $v^{dep}$ to Alice. On the other hand, if Alice fails to redeem $C^{col}_A$ before time $T$, she can still redeem it using $pre_A'$ after $T$, allowing Bob to refund $v^{dep}$. However, if Alice reveals $pre_A$ before $T$ but it is not included in the block, she will publish both $pre_A$ and $pre_A'$ to redeem $C^{col}_A$. In this case, if $C^{dep}$ reads both $pre_A$ and $pre_A'$ from $C^{col}_A$, it will permanently lock or burn $v^{dep}$.

Bob's bribery attack mentioned above doesn't make him gain more than what he would have earned by behaving honestly, i.e. if he bribe miners for not including $tx^{dep}_{pre_A}$ in the block, then after time $T$, Alice will redeem $C^{col}_A$ with both $pre_A$ and $pre_{A}^{'}$. So $C^{dep}$ will burn $v^{dep}$ and Bob end up earning $v^{col}_B - k*br$, where $br>f^{dep}_A$, and $k$ is the number of blocks mined between time $t_{pub}$ and $T$.

Furthermore, If Bob don't release $tx^{col}_B$, revealing $pre_B$, before $T$, he would earn $v^{col}_B - v^{ded}$, neither he gets any share in $v^{dep}$. Similarly, if Alice doesn't release $pre_A$, then he will end up loosing $v^{dep}$.


With this brief understanding of the proposed attacks and \prot\ next, we discuss the detailed analysis of the proposed M2MB attack on the He-HTLC and provide the game theoretic argument to show the gain that miners get by deviating from the protocol. Furthermore, we will provide the analysis of \prot.
\section{M2MBA analysis }
\label{sec:M2MBAttackAnalysis}
This section outlines our system model for analyzing the proposed M2MBA on the He-HTLC and our protocol. Following this, we will delve into the analysis of both.

\subsection{System Model}
The system model we adopt bears resemblance to the one employed in MAD-HTLC~\cite{Tsabary2021} and He-HTLC~\cite{SarishtWadhwa2022}. Our model considers a blockchain-based cryptocurrency that enables token transactions among a group of participants (entities). These participants include Alice, Bob, other users, and a predetermined group of $n$ miners, which we refer to as $M = {M_1, ...., M_n}$.


Each entity (including miners) in the system can create transaction that deploy a contract. The contract deposit tokens and defines the predicate on the token transfer. Furthermore, contract can output specific value, which the other contract can read. Each entity can create transaction that transfer token to other entities. A transaction transferring the token must supply input values such that the contract predicate evaluated over them is true, after which tokens are transferred to the specific entity defined in the contract, and contract is said to be {\em redeemed}. Transaction must satisfy specifically the following predicates, in addition to the other predicates: 1) digital signature $sig$ provided by the transaction matches a public key $pk$ specified in the contract, 2) preimage $pre$ provided by the transaction matches a hash digest $dig$ specified in the contract, i.e., that $H(pre) = dig$.


For the predicate to be imposed correctly we assumed that every participant has access to a digital signature scheme~\cite{Christian2017} with a security parameter of $\mu$ and a hash function, denoted as $H: \{0, 1\}^* \rightarrow \{0, 1\}^\mu$, which maps arbitrary-length inputs to outputs of length $\mu$. The security parameter $\mu$ has been selected such that the standard cryptographic assumptions are satisfied, which implies that the digital signature scheme is resistant to existential forgery  attacks~\cite{Christian2017, Goldwasser1988}, and the hash function $H$ is resistant to preimage attacks~\cite{Dziembowski2018}.

Invalid transactions with a negative predicate value will not be included in any block and will be ignored. For a transaction to be considered valid, the token transferred must not exceed the amount deposited in the contract at the time of deployment. The transaction fee is calculated as the difference between the transferred token amount and the deposited amount.

The created transactions are submitted to the miners. Miners add the received valid transaction in the data structure called {\em mempool}. The transaction in {\em mempool} are called as {\em unconfirmed} transactions, which miner can confirm by including them in the block. Miners extends the blockchain by creating the new block. Each miner $i$ can create a block with a specific rate denoted by $\lambda_i$, such that  $\sum_{n}^{i=1} \lambda_i = 1$. The block appended at the $j^{th}$ height is referred to as a $b_j$. We further assume that block-creation as a discrete-time, memory-less
stochastic process. We also assume that, only one miner mines a block at specific time instance, hence there doesn't exist a transient inconsistencies in our system.

Furthermore, We assume there are two types of miners: {\em active miners and passive miners}. 
Active miners proactively seek out other participants to participate in external protocols, say by deploying contracts, which we are going to see next section (\S\ref{sec:M2MBA}). By doing so, these miners can identify and capitalize on opportunities that maximize their utility in terms of the number of tokens earned. They do so by dividing the gained utilities among the colluding parties. Active rational miners are considered more formidable than their passive counterparts which only uses the transaction information available in the mempool to maximize their utility.

Out of $n$ miners, say $col$ miners are active miners. The consolidated mining power of the active miners is represented as $\lambda_{col}$. Active miners 
agree to share the rewards proportional to their individual mining power. The probability that a specific miner $M_i$ within the colluding group gets the reward, given that one of the colluding miners mines the block, is proportional to their mining power. In other words, the probability that miner $M_i$ receives a reward, given that one of the active colluding miners mines the block, is given by:
P(Miner $M_i$ receives reward/one colluding miner mines) $=\lambda_i/\lambda_{col}$, where P represents probability function.

For the resemblance with the production blockchain, we assume that there exists other users which generate the transaction that are not related to the \prot, we called such transaction as an {\em unrelated} transaction. Each transaction offers a fee to the miner for their inclusion in the block. We further assume that transaction related to the \prot\ offers a higher transaction fee compared to the {\em unrelated} transactions, unless specified otherwise. Since we assume that forks do not exist, it is assumed that a transaction is confirmed as soon as it is included in a block.


Beyond the previous assumptions, we model the system as a game among various entities, including miners. All entities are rational, risk-neutral, and non-myopic, with ties broken randomly. They aim to maximize expected utility, and no discounting is applied—i.e., a payment of $x$ tokens holds equal utility regardless of time. These assumptions are consistent with those in MAD-HTLC and He-HTLC.

\subsection{M2MBA on He-HTLC}
\label{sec:M2MBA}
In this section, we will delve into the M2MBA on the He-HTLC. 
Initially, we will discuss about the attack strategy of active miners. The basic idea is that some miner will deploy the contract $C_{M2M}$, and each miner $M_i$ will lock their collateral of $v^{col}$ in it. Locked collateral acts as a promise for getting bribed for censoring $tx^{dep}_A$. The miner who confiscates the $v^{col}$ following {\em col-M} will send the promised bribe to the other miners for censoring $tx^{dep}_A$. On the other hand, for passive miners, the attack strategy is simply to exclude $tx^{dep}_A$ from the block until time $T$ and include $tx^{col}_M$ as soon as Bob releases $tx^{dep}_B$. Next, we discuss the M2MB attack in detail.


\subsubsection{M2MBA Details: } Let’s say a He-HTLC is established between Bob and Alice, where they both deposits $v^{dep}$ and $v^{col}$, respectively. The attack proceed in three phases 
First, a miner deploy the contract $C_{M2M}$ represented in Algorithm~\ref{algo:M2MBA}, followed by which each miner, lock the amount of $v^{col}$ in the contract by calling {\em LockCollateral}. They can do this before or immediately after Alice publishes $tx^{dep}_A$. Second, each miner mining the block between time $t_{pub}$ to $T$ lock their bribe (bribe they expected to earn in future for censoring $tx^{dep}_A$) by calling {\em RequestBribe}. The idea behind this is that if miners are assured of receiving a profitable amount in the future, they are more likely to choose not to include $tx^{dep}_A$ in the block.

Lastly, immediately after time $T$ (i.e. at time $T+1$), when Bob discloses $pre_B$ via $tx^{dep}_B$, transferring $v^{dep}+v^{col}$ to He-Col, miners vie for its inclusion in the block. 
Note that since Bob is honest, he discloses $pre_B$ through $tx^{dep}_B$ at time $T+1$. Since every miner is aware of both $pre_A$ and $pre_B$, they also compete to add $tx^{col}_M$ to the same block and confiscate $v^{col}$.  
The first miner, say $M_{brb}$ who confiscate the $v^{col}$ from {\em He-Col} by following {col-M} (i.e. by adding $tx^{col}_M$ in the block)
will send the promised bribe to the miners who mine the blocks between time $t_{pub}$ to $T$ for censoring $tx^{dep}_A$. The amount locked in the first step for other miners, except $M_{brb}$, will be unlocked and sent back to them. However, the amount remains after paying bribe is sent to $M_{brb}$, depicted at line 41-43 in Algorithm~\ref{algo:M2MBA}.

\subsubsection{Game Details}
\label{sec:gamesetup}
We model the M2MBA as a game, $\zeta^{mba}$ between Alice, Bob, and miners $M = \{M_1, \cdots , M_n\}$. The game runs for $T+1$ 
rounds, where each round $j$ is represented by a block $b_j$.
Like in~\cite{SarishtWadhwa2022}, without loss of generality, we say the game begins when He-Dep and He-Col contracts are initiated in some block $b_0$. Furthermore, the round during which Alice releases $pre_A$ is identified as $t_{pub}$. 

In every round $k$, each party (Alice, Bob, and Miners) observes the following four game states: 1) {\em red}: {\em He-Dep} remains redeemable; 2) {\em nred-nrev}: He-Dep has already been redeemed and forwarded to He-col, but the miners are unaware of $pre_A$; 3) {\em nred-rev}: {\em He-Dep} has already been redeemed and dispatched to He-col, and some miners are cognizant of $pre_A$, and 4) {\em nred-A}: He-Dep has already been redeemed with $pre_A$. We denote this states set as $S=${{\em red, nred-nrev, nred-rev, nred-A}}.

Each round, $k \in [ t_{pub}, T]$, corresponds to the subgame $\zeta^{mba}(k, s)$, where $s$ represents the state and denotes that $T-k$ more blocks must be produced following this subgame. Hence we look the game as the series of subgames. We use "." as a placeholder when referring to subsets of games, such as $\zeta^{mba}(.,red)$, which denotes the collection of all subgames in which {\em He-Dep} can still be redeemed. Within each subgame, the participating parties may execute certain {\em actions}, which have the potential to alter the game state. Next we will discuss the set of action that each participating entity can take.

Alice complies with the He-HTLC protocol, whereby she has the option to select a round $t_{pub}$ to release $tx^{dep}_A$, offering a fee $f^{dep}_A$ of her preference. Once Alice broadcasts $tx^{dep}_A$, i.e., exposes $pre_A$, and the miner includes $tx^{dep}_A$ in the chain by mining the block, the protocol is completed. Bob's actions do not impact the protocol execution since exposing $tx^{dep}_A$ transfers $v^{dep}$ to Alice and $v^{col}$ to Bob. Alternatively, in the scenario where $tx^{dep}_A$ is not disclosed on-chain, $tx^{dep}_B$ revealed by Bob will be placed on-chain after time $T$. This transaction will transfer the sum of $v^{dep}$ and $v^{col}$ to the {\em He-Col}, which will ultimately (after $l$ rounds) be transferred to Bob if $pre_A$ is not revealed, either by honest or malicious behaviour from Alice.


The actions of miners involve strategically selecting transactions to be mined in order to maximize their utility. The transactions available for miners to choose from are dependent on the current round $k$ as well as their knowledge of $pre_A$ and $pre_B$. In any given subgame $\zeta^{mba}(., .)$, a miner $M_i$ can include unrelated transactions from the mempool and earn a fee of $f$. However, in subgames $\zeta^{mba}(k, red)$, where $T \geq k \geq t_{pub}$, $M_i$ has the option to include $tx^{dep}_A$ for a fee of $f^{dep}_A$ or censor it with the expectation of the higher profit ($>f^{dep}_A$) in future. Similarly, in subgames $\zeta^{mba}(k, red)$ where $k > T$, $M_i$ can choose to include $tx^{dep}_B$ for a fee of $f^{dep}_B$, but only if $tx^{dep}_A$ has not already been included. 

Furthermore, miners action is limited to include only the unrelated transactions to maximize their utility in the subgame $G^{dep}(k$,\\$\text{\em nred-nrev})$.
However, in subgames $G^{dep}(k,\text{\em nred-rev})$, where $k > T$, $M_i$ can create and include a transaction $tx^{col}_M$ that redeems He-Col for itself via path {\em col-M}. Furthemore, as soon as Alice reveals $pre_A$, a miner can decide to pay a pre-agreed $br$ to other miners in exchange for allowing him
to redeem He-col in some future block (after round $T$). 
For ease of exposition, we assume that miners have reached an agreement about the
value of $br$ before or immediately after $t_{pub}$. Miner can independently decide whether or not to accept each bribe $br$, however, each miner receives only one of these bribes, i.e., the bribe from the miner who redeem {\em He-Col}. By undertaking these actions, each participant in the game accumulates certain utilities, which we will delve into in the following discussion.

Each entity's utility, $u_i$, in the game is determined by the tokens they earn at the end of the game, i.e., after $T$ blocks have been created. Specifically, for each entity $i$, the utility function $u_i: \text{Action X } (\mathbb{Z}, \text{States}) \rightarrow \mathbb{R}$ maps to the number of tokens entity $i$ earns at the end of the game. The utility of player $i$ when action $a$ is taken in the game $\zeta^{mba}$ is denoted as $u_i(a, \zeta^{mba})$. The maximum utility that player $i$ can achieve in game $\zeta^{mba}$ is referred to as $u^{max}_i(\zeta^{mba})$. We assume that the utility obtained from a block containing only unrelated transactions is $f$. Therefore, if a transaction related to {\em He-HTLC} with utility $x$ is included in exchange for an unrelated transaction, it would have a utility of $x - f$.

Given above, our game-theoretic analysis demonstrates that it is a dominant strategy for a miner, $M_j$, is to accept a bribe from a bribing miner, if he couldn't redeem {\em He-Col}. Else the dominant strategy is to redeem {\em He-Col} than following the {\em He-HTLC} protocol.


\subsubsection{Analysis} 
\label{sec:M2MBADerailedAnalysis}
As discussed earlier, M2MBA proceeds in three phases. In the first (setup) phase, miners agree on the bribe $br$ and each active miner locks a collateral of value $v^{col}$. Let $\lambda_{col}$ denote the mining power of miners who locked their collateral via {\em LockCollateral}.

Given the above game setup, we prove that for all miners, censoring $tx^{dep}_A$ and accepting the bribe is a dominant strategy when $t \leq T$. Lemma~\ref{lm:2} shows that at round $T+1$, invoking He-Col to confiscate $v^{col}$ is also dominant. Thus, attacking dominates honest participation (Theorem~\ref{thm:M2MBA}). Lemma proofs are deferred to Appendix~\ref{sec:M2MBAproof}.


\begin{lemma}
\label{lm:1}
In $\zeta^{mba}(k, red)$, where $k \leq T$, if $(T-t_{pub})*br*\lambda_i/\lambda_{col}>f^{dep}_A$, then it is a dominant move for active miners to accept the bribe by censoring the transaction $tx^{dep}_A$ instead of including it in the block.
\end{lemma}


\begin{lemma}
\label{lm:2}
In $\zeta^{mba}(k, red)$, where $k \leq T$, if $v^{col}*\lambda_i/\lambda{col} + (T-t_{pub})*br*(2*\lambda_i/\lambda_{col} - 1)>f^{dep}_A$, then it is a dominant move for active miners to offer a bribe to other miners to censor the transaction $tx^{dep}_A$ and perform the attack rather than including it in the block. 
\end{lemma}

\begin{lemma}
\label{lm:3}
        In $\zeta^{mba}(k, red)$, where $k \leq T$. If $v_{col}*\lambda_i > f^{dep}_A$, it is the dominant move for passive miners to censor $tx^{dep}_A$ rather than include it in the block.
\end{lemma}

\begin{lemma}
\label{lm:4}
 In $\zeta^{mba}(k, red)$, where $k > T$. Let $z$ be the net earning of the miner by confiscating $v^{col}$. If $z > f^{dep}_A$, and Bob release transaction $tx^{dep}_B$ and hence release secret $pre_B$ at $T+1$, it is the dominant move for all miners to mine the block that includes $tx^{col}_M$, at $k=T+1$. (transient transition from $\zeta^{mba}(k, red)$ to $\zeta^{mba}(k, nred$-$rev$)  by including $tx^{dep}_B$ and then including $tx^{col}_M$ in the same block)
\end{lemma}

\begin{lemma}

    We made the valid assumption that for an active miner, inequality $v^{col}\lambda_i/\lambda_{col} + (T - t_{pub})*br*(2\lambda_i/\lambda_{col} - 1) > f^{dep}_A$ holds. Similarly, for a passive miner, Assumption $v^{col}*\lambda_i > f^{dep}_A$ is also valid. 
\end{lemma}

\begin{theorem}
\label{thm:M2MBA}
     It is dominant strategy for all the rational miners to perform the M2MBA than following the He-HTLC protocol, if $v^{col}*\lambda_i > f^{dep}_A$ and $v^{col}\lambda_i/\lambda_{col} + (T - t_{pub})*br*(2\lambda_i/\lambda_{col} - 1) > f^{dep}_A$ . 
\end{theorem}
\begin{proof}
    
    We will divide the proof into two cases: the behavior of passive miners and that of active miners.

{\em Passive Miners:} The potential strategies for passive miners are as follows:
a) Include $tx^{dep}_A$ in the block.
b) Wait until round $T+1$, when Bob releases $tx^{dep}_B$, and include both $tx^{dep}_B$ and $tx^{dep}_M$ in the same block in round $T+1$ (this constitutes an attack).
c) Act honestly and include only $tx^{dep}_B$.
Given that $f^{dep}_A = f^{dep}_B$, under choices a) and c), the miner would earn $f^{dep}_A$. However, according to Lemma~\ref{lm:3} and Lemma~\ref{lm:4}, choice b) dominates both a) and c), making the attack the most profitable strategy for passive miners.

{\em Active Miners:} The possible moves for active miners are:
a) Include $tx^{dep}_A$ in the block.
b) Accept a bribe from other active miners to censor $tx^{dep}_A$.
c) Offer a bribe to other active miners to censor $tx^{dep}_A$ until round $T+1$, when Bob releases $tx^{dep}_B$, and then include both $tx^{dep}_B$ and $tx^{dep}_M$ in the same block.
d) Act honestly and include only $tx^{dep}_B$.
Again, assuming $f^{dep}_A = f^{dep}_B$, following options a) and d), the miner would earn $f^{dep}_A$. However, from Lemma~\ref{lm:1}, option b) dominates both a) and d). Furthermore, from Lemma~\ref{lm:2} and Lemma~\ref{lm:4}, option c) also dominates a) and d). It is important to note that by accepting the bribe, the miners participate in the attack.

In conclusion, based on the analysis above, performing the attack is the dominant strategy for both passive and active miners, rather than following He-HTLC."

    

    
\end{proof}

\section{\prot\ Analysis} 
\label{sec:DEMBAAnalysis}
In this section, we present a detailed analysis of \prot. The game setup remains consistent with that outlined in the previous section (Section~\ref{sec:gamesetup}). 
Let $S$ represent the set of state of the game, where $S=\{all\text{-}red, nred\text{-}AB, nred\text{-}A'B,$\\ $nred\text{-}AA'B, nred\text{-}ABT, nred\text{-}A'BT, nred\text{-}AA'BT\}$. The description of each state is as follows: 1) $all\text{-}red$ - indicates that all smart contracts are redeemable, 2) $nred\text{-}AB$ - $C^{col}_A$ is redeemed with $pre_A$ and $C^{col}_B$ is redeemed before $T$, 3) $nred\text{-}A'B$ - $C^{col}_A$ is redeemed with $pre_A'$ and $C^{col}_B$ is redeemed before $T$, 4) $nred\text{-}AA'B$ - $C^{col}_A$ is redeemed with both $pre_A$ and $pre_A'$, and $C^{col}_B$ is redeemed before $T$. Similarly, $nred\text{-}ABT$, $nred\text{-}A'BT$, $nred\text{-}AA'BT$ follow the same description except $C^{col}_B$ is redeemed after $T$. 
 
Given this, we demonstrate that under \prot, no party-including miners-can profit by deviating from the protocol; in the state $nred\text{-}AB$, all involved parties earn more by adhering to it. Due to space constraints, all lemma proofs are provided in the appendix~\ref{sec:DEMBAproof}.


\begin{lemma}
     \label{lemma:AliceDominant}
     For $v^{ded}>0$, it is dominant move for the Alice to follow \prot\ than deviating from it.
\end{lemma}

\begin{lemma}
     \label{lemma:BobDominant}
     For $v^{ded}>0$, it is dominant move for the Bob to follow \prot\ than deviating from it.
\end{lemma}


\begin{lemma}

 \label{lemma:MinerDominant}
     For $\alpha<1$, it is dominant move for the miners to follow \prot\ than deviating from it.
\end{lemma}

\begin{theorem}
    It is dominant strategy for all the parties (Alice, Bob, and Miners) to follows the protocol than deviating from it.
\end{theorem}

\begin{proof}
From Lemmas~\ref{lemma:AliceDominant}, \ref{lemma:BobDominant}, and \ref{lemma:MinerDominant}, the intersection of states where each party maximizes its earnings reveals that $nred\text{-}AB$ is the common state offering high rewards for all. Thus, adhering to \prot\ emerges as the dominant strategy for Alice, Bob, and the miners, since any deviation yields lower payoffs.

Next, we demonstrate that no two parties can collaborate to earn more than they would by following \prot. There are three possible collaborations: (1) Alice-Bob, (2) Alice-Miner, and (3) Bob-Miner. Collaboration between Alice and Bob can be ruled out since they are counter-parties in the transaction.

If Alice behaves honestly, her expected earnings are  $v^{dep}+v_A^{col}$. Since no additional tokens are created in the exchange process (except for mining rewards granted by the consensus protocol), Alice cannot earn more than $v^{dep}+v_A^{col}$ even by collaborating with the Miner. Similarly, the Miner cannot earn more than the transaction fee and mining reward, which they are already entitled to by behaving honestly.

Finally, consider collaboration between Bob and the Miner. In this case, Alice retains the deciding move on whether $v^{dep}$ is returned to Bob. Thus, Bob cannot earn more than $v_B^{col}$, the amount he would receive by behaving honestly. Consequently, collaboration offers no advantage to Bob or the Miner.
 
\end{proof}

\section{\prot\ Evaluation}
\label{sec:DEMBAEvaluation}
In this section, we present the evaluation of \prot\ in comparison to state-of-the-art approaches. We begin by outlining the implementation details, followed by a description of the evaluation setup. Finally, we conclude with a discussion of the evaluation results.

\subsection{Implementation Details}
We demonstrate the effectiveness of \prot\ through its evaluation on both Bitcoin and Ethereum. While deploying \prot\ on smart contract-enabled blockchains like Ethereum and Hyperledger Fabric is relatively straightforward, this section focuses primarily on detailing our implementation of \prot\ for Bitcoin.

The Bitcoin scripting language lacks a built-in mechanism for storing and retrieving arbitrary data, such as interacting with variables or state, as is possible with Ethereum smart contracts. This limitation affects the communication between $C^{dep}$, $C^{col}_A$, and $C^{col}_B$. In Bitcoin, each smart contract is represented as a UTXO, with respective names denoted as $UTXO^{C^{dep}}$, $UTXO^{C^{col}_A}$, and $UTXO^{C^{col}_B}$.

\begin{figure*}
\centering
\includegraphics[height=5.15cm, width=0.7\linewidth]{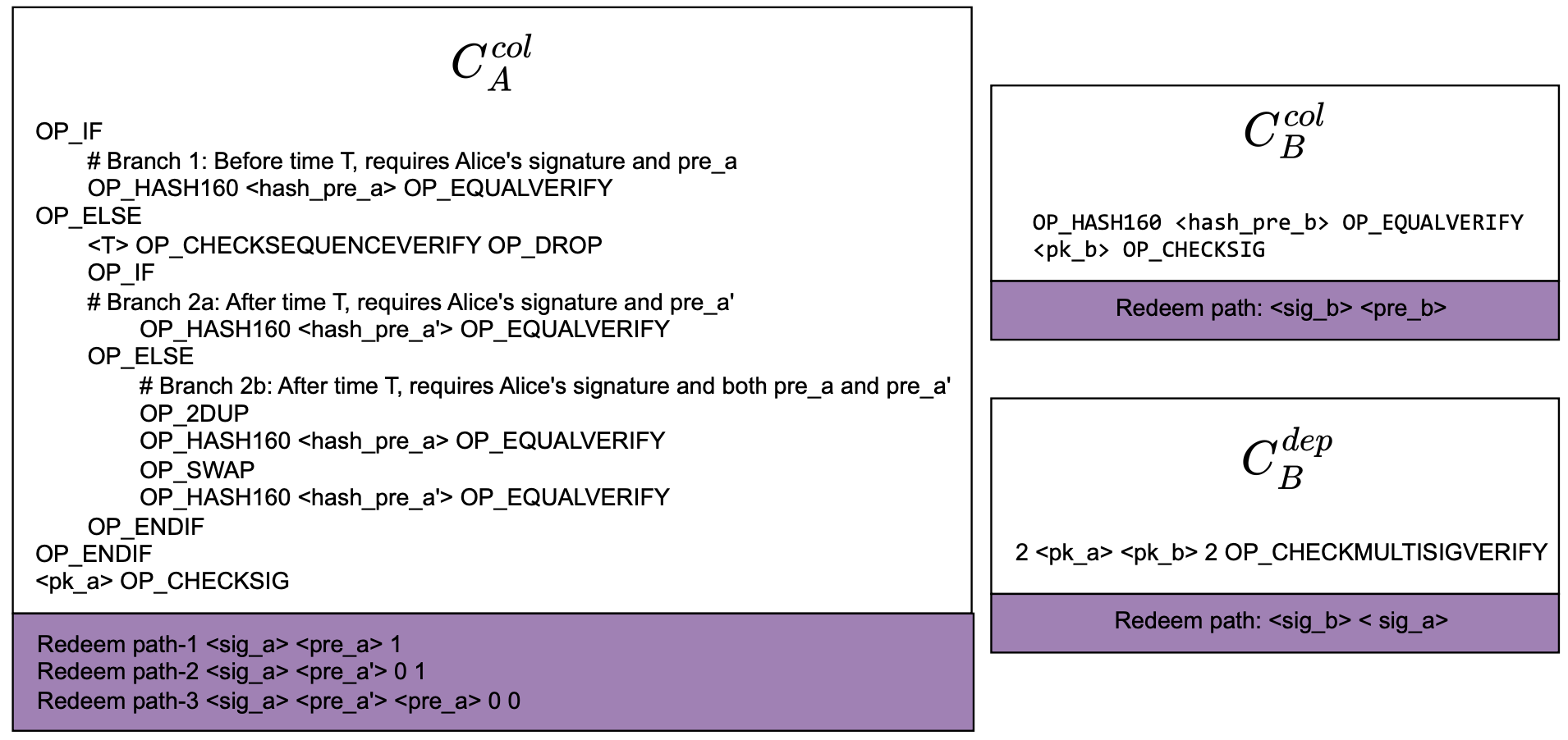}
\caption{Bitcoin implementation (locking script) for \prot\ contracts}
\label{fig:contractImplementation}
\end{figure*}

Alice creates $UTXO^{C^{col}_A}$, while Bob creates $UTXO^{C^{dep}}$ and \\ $UTXO^{C^{col}_B}$. The Bitcoin locking scripts for these UTXOs are illustrated in Figure~\ref{fig:contractImplementation}. As described earlier, $UTXO^{C^{col}_A}$ can be redeemed through three distinct paths: 
\begin{itemize}
    \item Using $pre_A$ via the transaction $tx^{col}_{pre_A}$, 
    \item Using $pre_A'$ via the transaction $tx^{col}_{pre_A'}$, and 
    \item Using both $pre_A$ and $pre_A'$ via the transaction $tx^{col}_{pre_{AA}'}$.
\end{itemize}
Each of these transactions must be signed by Alice.

In contrast, $UTXO^{C^{col}_B}$ has only one redeem path, which involves revealing $pre_B$ through the transaction $tx^{col}_{pre_B}$, signed by Bob.

The final set of transactions transfers the funds to the respective parties by redeeming $UTXO^{C^{dep}}$. These transactions depend on how $UTXO^{C^{col}_A}$ is redeemed, leading to three possible paths:
\begin{itemize}
    \item $tx^{dep}_{pre_A}$, which uses the outputs of $tx^{col}_{pre_A}$ and $tx^{col}_{pre_B}$,
    \item $tx^{dep}_{pre_A'}$, which uses the outputs of $tx^{col}_{pre_A'}$ and $tx^{col}_{pre_{B}}$, and
    \item $tx^{dep}_{pre_{AA}'}$, which uses the outputs of $tx^{col}_{pre_{AA}'}$ and $tx^{col}_{pre_{B}}$.
\end{itemize}
All these transactions require signatures from both Alice and Bob. This structured transaction mapping ensures the correct ordering of events, as each $tx^{dep}$ transaction takes inputs from the relevant $tx^{col}$ transactions. A detailed graph illustrating this process is provided in Figure~\ref{fig:transactionGraph}.

\begin{figure}
\centering
\includegraphics[height=5.05cm, width=0.69\linewidth]{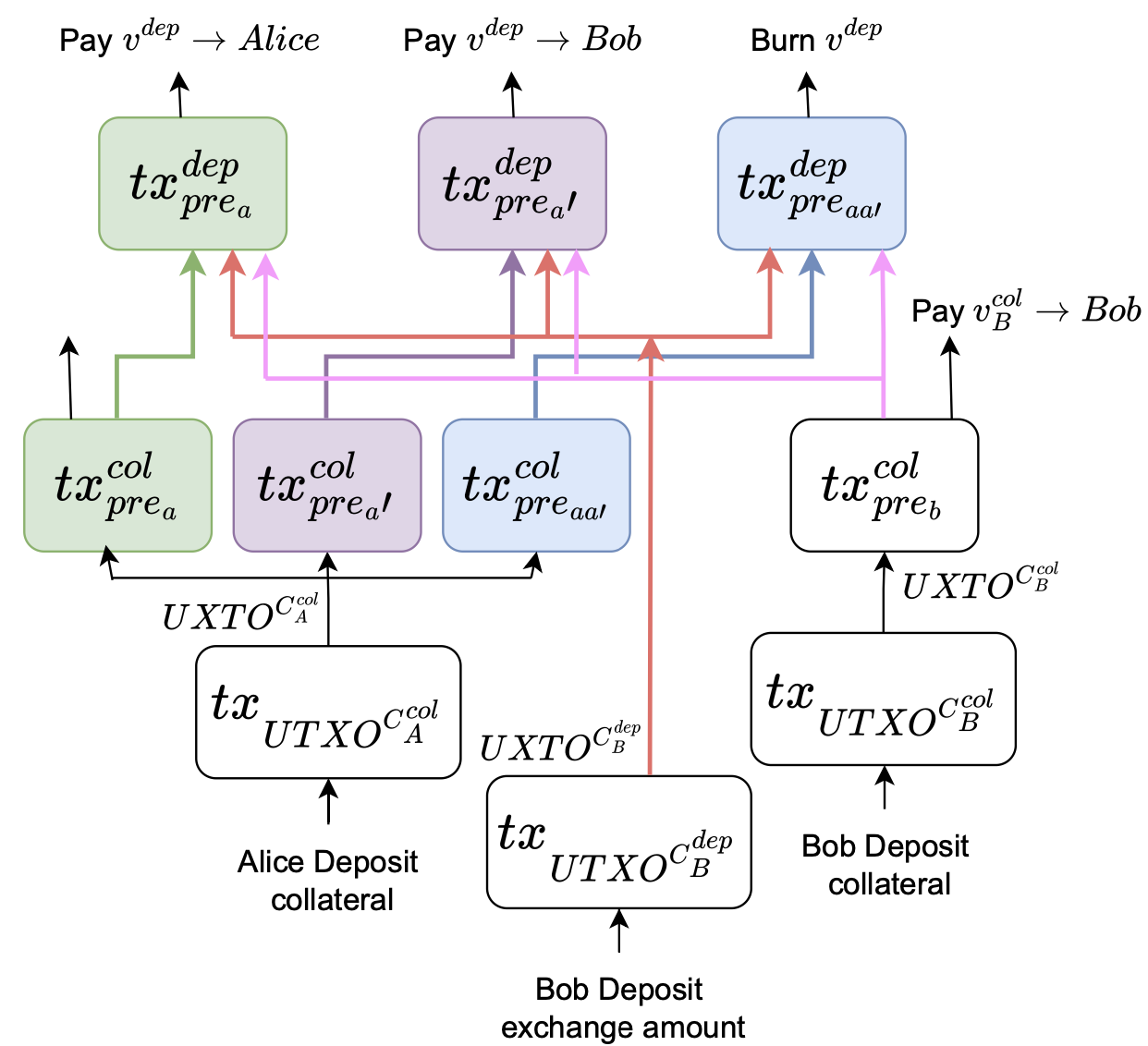}
\caption{Transaction dependency with respect to UTXO in Bitcoin implementation of \prot}
\label{fig:transactionGraph}
\end{figure}

This design addresses Bitcoin's inability to natively pass data between transactions, unlike Ethereum. Although Bitcoin provides the OP\_RETURN opcode for embedding arbitrary data in transactions, Bitcoin Script cannot directly read such outputs. Instead, external mechanisms, such as Transaction Indexing Services or Custom Off-Chain Logic, can be employed to achieve this functionality.

The following steps outline the protocol that each party follows:

\begin{itemize}
    \item Alice creates $UTXO^{C^{col}_A}$ along with the transactions $tx^{col}_{pre_A}$, $tx^{col}_{pre_A'}$, and $tx^{col}_{pre_{AA}'}$, and shares these transactions with Bob. Note that Alice does not share the signed versions of these transactions with Bob.

    \item Bob creates $UTXO^{C^{col}_B}$ and $tx^{col}_{pre_B}$, then shares the unsigned version of $tx^{col}_{pre_B}$ with Alice. 
    The first two steps are essential to verify that the construction of all $tx^{dep}$, as described in the following step, is correct.

    \item Bob further creates $UTXO^{C^{dep}}$ and the transactions $tx^{dep}_{pre_A}$, $tx^{dep}_{pre_A'}$, and $tx^{dep}_{pre_{AA}'}$. He signs these transactions and shares them with Alice. 
    However, Bob keeps $UTXO^{C^{dep}}$ secret at this stage, i.e not publish it onto the blockchain.

    \item Alice signs $tx^{dep}_{pre_A}$, $tx^{dep}_{pre_A'}$, and $tx^{dep}_{pre_{AA}'}$, and sends the signed versions back to Bob, ensuring that both parties possess duly signed transactions.

    \item After receiving the signed versions of $tx^{dep}_{pre_A}$, $tx^{dep}_{pre_A'}$, and $tx^{dep}_{pre_{AA}'}$, Bob publishes $UTXO^{C^{dep}}$. This step is crucial because if Bob publishes $UTXO^{C^{dep}}$ before obtaining Alice's signatures for $tx^{dep}_{pre_A'}$ (or the other transactions), his tokens may remain locked in $C^{dep}$ (i.e., in $UTXO^{C^{dep}}$ in a UTXO-based model) indefinitely if Alice withholds her signatures.
\end{itemize}

Once the above steps are done, Alice will send the transaction to the Blockchain, i.e., it send to a miner, which will broadcast it to the network. Next, we will discuss our network setup.

Furthermore, we have also implemented and evaluate the \prot\ for Ethereum. For the ease of implementation\footnote{You can find the Solidity implementation of \prot\ here \url{https://github.com/nitinawathare14/DEMBAEth.git}}, we haven't implemented signature verification in it. Moreover to make the comparison fair, we didn't implement the signature verification for He-HTLC and MAD-HTLC as well, with which we have compared \prot.

\subsection{Evaluation Setup}
\label{sec:evaluationSetup}
We set up a private blockchain network consisting of 50 Oracle Virtual Machines to measure the protocol completion time—i.e., the time required to successfully complete the exchange—in a real network environment with potential blockchain forks \footnote{It is worth noting that, for comparing \prot\ in terms of transaction costs, deploying and sending transactions on a single-node blockchain network is sufficient.}. In this setup, each VM is provisioned with a 2.19 GHz dual-core CPU, 8 GB RAM, and a 128 GB HDD. All VMs run Ubuntu 20.04 and are configured with download and upload bandwidths of 1 Gbps and 100 Mbps, respectively. Each VM hosts a single blockchain node.

Throughout the experiment, we control the block mining difficulty to maintain an average block generation time of 15 seconds for Ethereum and 10 minutes for Bitcoin, reflecting the time required to solve the Proof-of-Work (PoW) puzzle. As a result, the average block inter-arrival time in our experiment equals the PoW solving time (15 seconds or 10 minutes), plus the block propagation delay, block creation time, and block validation time. While Ethereum 2.0~\cite{ethereum2docs}, based on Proof-of-Stake (PoS), was already released at the time of our experiments, we deliberately focus on the earlier Proof-of-Work (PoW) version to ensure a consistent comparison with Bitcoin.

We allocate mining power to each node in our 50-node setup based on the mining power distribution of the top 50 miners in the Ethereum network~\cite{ethMining}. Collectively, these miners account for approximately 99.98

Additionally, we replicate the geographical distribution of these miners by assigning our nodes locations consistent with the data from~\cite{ethMining}. We configure inter-node latencies using Linux’s ${\tt tc}$ tool to match the real-world ping delays corresponding to these geographic locations~\cite{pingDelay}, effectively simulating the network delays of Ethereum’s mainnet.

For network topology, inspired by the Bitcoin network-where node degrees follow a power-law distribution~\cite{bitcoinTopology}-we design our experiment such that each node connects to a random subset of peers, ensuring the degree distribution adheres to a power-law.


\subsection{Evaluation Result}
In this section, we present the evaluation results. Specifically, we use two key metrics—transaction cost and time to completion—to compare the performance of \prot\ against state-of-the-art protocols.


The results for Bitcoin are shown in Figure~\ref{fig:transSizeComparison}, where \prot\ reduces Alice's redemption cost by 42\% compared to MAD-HTLC but incurs 17\% higher cost than He-HTLC, as it requires two transactions instead of one. In contrast, for Bob’s refund, \prot\ lowers the cost by 42\% relative to He-HTLC and by 150\% compared to MAD-HTLC. Specifically, the transaction size in \prot\ is 48 bytes smaller than in He-HTLC and 294 bytes smaller than in MAD-HTLC.

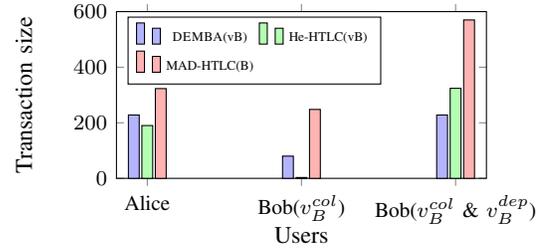
\begin{figure}[t!]
    \centering
    \pgfplotsset{footnotesize,height=3.8cm, width=0.84\linewidth}
    \begin{tikzpicture}
    \begin{axis}[
        ybar,
        bar width=4.1pt,
        /pgfplots/ybar=1pt,
        legend pos=north west,
        legend columns=2,
         legend style={font=\tiny},
        ymin=0,
        ymax=600,
        width=6.5cm,
        symbolic x coords={Alice, Bob($v^{col}_B$), Bob($v^{col}_B$ \& $v^{dep}_B$)},
         xticklabel style={name=tick no \ticknum},
        xtick=data,
        xlabel={Users},
        ylabel= {Transaction size}
        ]

        \addplot[fill=blue!30!white] table [x=users, y=DEMBA, col sep=comma] {data/gethPeerSimilarity.csv};
        \addplot[fill=green!30!white] table [x=users, y=He-HTLC, col sep=comma] {data/gethPeerSimilarity.csv};
        \addplot[fill=red!30!white] table [x=users, y=MAD-HTLC, col sep=comma] {data/gethPeerSimilarity.csv};
        
        \addlegendentry{\prot(vB)}
        \addlegendentry{He-HTLC(vB)}
        \addlegendentry{MAD-HTLC(B)}
    \end{axis}
    \end{tikzpicture}
    \caption{Comparison of \prot, He-HTLC, and MAD-HTLC with respect to the transaction cost (size) on the Bitcoin Network}
    \label{fig:transSizeComparison}
\end{figure}




Figure~\ref{fig:transSizeComparisonEthereum} presents the Ethereum results, showing trends similar to Bitcoin but with smaller margins. \prot\ reduces Alice’s claim cost by 38\% compared to MAD-HTLC, but it is 7\% costlier than He-HTLC, as it requires two transactions instead of one. For Bob’s refund, \prot\ lowers the cost by 23\% vs. He-HTLC and 71\% vs. MAD-HTLC. Furthermore, upon closer analysis, we conclude that the observed results stem from the same reason—\prot\ requires Alice to send two transactions to redeem $v^{dep}$, whereas He-HTLC requires only one.




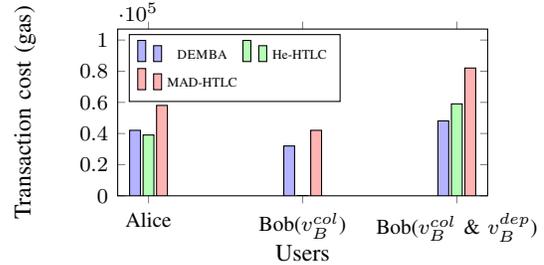
\begin{figure}[t!]
    \centering
    \pgfplotsset{footnotesize,height=3.8cm, width=0.84\linewidth}
    \begin{tikzpicture}
    \begin{axis}[
        ybar,
        bar width=4.1pt,
        /pgfplots/ybar=1pt,
        legend pos=north west,
        legend columns=2,
        legend style={font=\tiny},
        ymin=0,
        ymax=107060,
        width=6.5cm,
        symbolic x coords={Alice, Bob($v^{col}_B$), Bob($v^{col}_B$ \& $v^{dep}_B$)},
         xticklabel style={name=tick no \ticknum},
        xtick=data,
        xlabel={Users},
        ylabel= {Transaction cost (gas)}
        ]

        \addplot[fill=blue!30!white] table [x=users, y=DEMBA-ETH, col sep=comma] {data/gethPeerSimilarity.csv};
        \addplot[fill=green!30!white] table [x=users, y=He-HTLC-ETH, col sep=comma] {data/gethPeerSimilarity.csv};
        \addplot[fill=red!30!white] table [x=users, y=MAD-HTLC-ETH, col sep=comma] {data/gethPeerSimilarity.csv};
        
        \addlegendentry{\prot}
        \addlegendentry{He-HTLC}
        \addlegendentry{MAD-HTLC}
    \end{axis}
    \end{tikzpicture}
    \caption{Comparison of \prot, He-HTLC, and MAD-HTLC with respect to the transaction cost (gas used) on the Ethereum Blockchain Network}
    \label{fig:transSizeComparisonEthereum}
\end{figure}

\begin{figure*}[t!]
    \centering
    \begin{subfigure}{0.31\linewidth}
    \pgfplotsset{footnotesize,height=3.8cm, width=0.60\linewidth}
    \begin{tikzpicture}
    \begin{axis}[
        ybar,
        bar width=4.1pt,
        /pgfplots/ybar=1pt,
        legend pos=north west,
        legend columns=2,
        legend style={font=\tiny},
        ymin=0,
        ymax=10,
        width=6.0cm,
        symbolic x coords={$v^{dep}=v^{col}$, $v^{dep}=2v^{col}$,  $v^{dep}=4v^{col}$},
         xticklabel style={name=tick no \ticknum},
        xtick=data,
        xlabel={Variation in $v^{dep}$},
        ylabel= {Time to complete (minutes)}
        ]


        \addplot[fill=blue!30!white] [error bars/.cd, y explicit,y dir=both,] table [x=users, y=DEMBAAlice, y error=DEMBAAliceCfi, col sep=comma]{data/timetocomplete.csv};

        \addplot[fill=green!30!white] [error bars/.cd, y explicit,y dir=both,] table [x=users, y=HeHTLCAlice, y error=HeHTLCAliceCfi, col sep=comma] {data/timetocomplete.csv};
        \addplot[fill=red!30!white] [error bars/.cd, y explicit,y dir=both,] table [x=users, y=MDHTLCAlice, y error=MDHTLCAliceCfi, col sep=comma] {data/timetocomplete.csv};

        \addlegendentry{\prot}

        \addlegendentry{He-HTLC}
        \addlegendentry{MAD-HTLC}
    \end{axis}
    \end{tikzpicture}
    \caption{Variation in protocol completion time when Alice redeem $v^{dep}$}
    \label{fig:Aliceredeem}

    \end{subfigure}\hfill
    \begin{subfigure}{0.31\linewidth}\pgfplotsset{footnotesize,height=3.8cm, width=0.60\linewidth}
    \begin{tikzpicture}
    \begin{axis}[
        ybar,
        bar width=4.1pt,
        /pgfplots/ybar=1pt,
        legend pos=north west,
        legend columns=2,
        legend style={font=\tiny},
        ymin=0,
        ymax=10,
        width=6.0cm,
        symbolic x coords={$v^{dep}=v^{col}$, $v^{dep}=2v^{col}$,  $v^{dep}=4v^{col}$},
         xticklabel style={name=tick no \ticknum},
        xtick=data,
        xlabel={Variation in $v^{dep}$},
        ]

        \addplot[fill=blue!30!white] [error bars/.cd, y explicit,y dir=both,] table [x=users, y=DEMBABob, y error=DEMBABobCfi, col sep=comma] {data/timetocomplete.csv};

        \addplot[fill=green!30!white] [error bars/.cd, y explicit,y dir=both,] table [x=users, y=HeHTLCBob, y error=HeHTLCBobCfi, col sep=comma] {data/timetocomplete.csv};
        \addplot[fill=red!30!white] [error bars/.cd, y explicit,y dir=both,] table [x=users, y=MDHTLCBob, y error=MDHTLCBobCfi, col sep=comma] {data/timetocomplete.csv};

        \addlegendentry{\prot}

        \addlegendentry{He-HTLC}
        \addlegendentry{MAD-HTLC}
    \end{axis}
    \end{tikzpicture}
    \caption{Variation in protocol completion time when Bob redeem $v^{col}$ (or $v^{col}_B$)}
    \label{fig:Bobredeem}
    \end{subfigure}\hfill
    \begin{subfigure}{0.31\linewidth}\pgfplotsset{footnotesize,height=3.8cm, width=0.60\linewidth}
    \begin{tikzpicture}
    \begin{axis}[
        ybar,
        bar width=4.1pt,
        /pgfplots/ybar=1pt,
        legend pos=north west,
        legend columns=2,
        legend style={font=\tiny},
        ymin=0,
        ymax=27,
        width=6.0cm,
        symbolic x coords={$v^{dep}=v^{col}$, $v^{dep}=2v^{col}$,  $v^{dep}=4v^{col}$},
         xticklabel style={name=tick no \ticknum},
        xtick=data,
        xlabel={Variation in $v^{dep}$},
        ]

        \addplot[fill=blue!30!white] [error bars/.cd, y explicit,y dir=both,] table [x=users, y=DEMBABobBoth, y error=DEMBABobBothCfi, col sep=comma] {data/timetocomplete.csv};

        \addplot[fill=green!30!white] [error bars/.cd, y explicit,y dir=both,] table [x=users, y=HeHTLCBobBoth, y error=HeHTLCBobBothCfi, col sep=comma] {data/timetocomplete.csv};
        \addplot[fill=red!30!white] [error bars/.cd, y explicit,y dir=both,] table [x=users, y=MDHTLCBobBoth, y error=MDHTLCBobBothCfi, col sep=comma] {data/timetocomplete.csv};

        \addlegendentry{\prot}

        \addlegendentry{He-HTLC}
        \addlegendentry{MAD-HTLC}
    \end{axis}
    \end{tikzpicture}
    \caption{Variation in protocol completion time when Bob redeems  $v^{col}$ ( or $v^{col}_B$) and $v^{dep}$}
    \label{fig:BobredeemBoth}

    \end{subfigure}
    
    \caption{Comparison of \prot, HTLC,  He-HTLC, and MAD-HTLC with time to completion on the Bitcoin Blockchain Network}
    \label{fig:timeToCompletion}
\end{figure*}
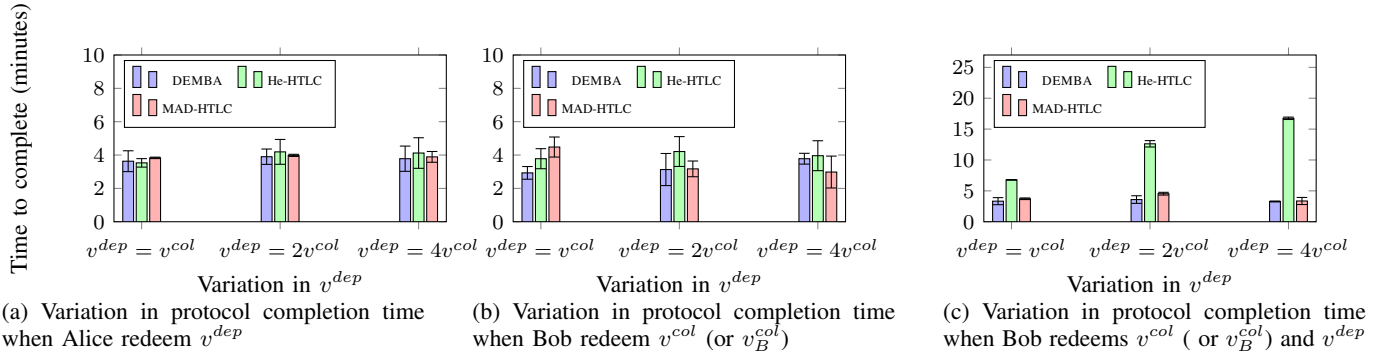


We further compared \prot\ with He-HTLC and MAD-HTLC in terms of time to complete (TTC), defined as the time from transaction initiation to final transfer. Results were averaged over 5 different runs, using $f^{dep}_A = f^{dep}_B = 0.0001 \times v^{dep}$, and reported with 95\% confidence intervals.

Our observations are as follows: when Alice redeem $v^{dep}$ and Bob redeem  $v^{col}$ (or $v^{col}_B$ in the case of \prot), the TTC for all protocols, including \prot, shows negligible variation with changes in the exchange amount $v^{dep}$, as depicted in Figure~\ref{fig:Aliceredeem} and~\ref{fig:Bobredeem}  However, when Bob redeems both $v^{dep}$ and $v^{col}$, all protocols except He-HTLC continue to show negligible variation in TTC. He-HTLC, on the other hand, exhibits an increase in completion time because its parameter $l$ depends on $v^{dep}$. Specifically, to satisfy the condition $v^{col} \ge \left( \dfrac{v^{dep}}{\kappa - 1} \right) + f$, an increase in $v^{dep}$ necessitates a corresponding increase in $\kappa$, which is directly proportional to the parameter $l$. Figure~\ref{fig:BobredeemBoth} demonstrates this behavior.

From this, we conclude that \prot's performance in terms of protocol completion time is on par with, if not better than, the existing state-of-the-art protocols under consideration. Moreover, \prot\ achieves this while incurring lower costs compared to them.

\section{Related Work}
\label{sec:relatedwork}

This section reviews related work in three areas: bribery attacks and countermeasures, mechanisms for fair exchange exploited in our approach, and strategies for optimizing transaction placement to maximize incentives—an alternative bribery attack.

\noindent{\bf Bribery Attacks: }
There have been proposals for bribery attacks aimed specifically at HTLC. In \cite{Winzer2019}, temporary censorship attacks are suggested where attackers can censor transactions from a victim until a timeout, potentially allowing Bob to censor Alice's transaction and reverse the payment. However, their attack focus primarily on collaboration of one of the parties involved in the exchange to the consensus nodes (miners).


The incentive compatibility of blockchain protocols has received extensive research attention. For example, the selfish mining attack~\cite{eyal2013majority} demonstrates that Bitcoin is not incentive-compatible, and miners could gain more by selectively withholding blocks.
Bonneau et al.\cite{Bonneau2016} demonstrate how miners can be bribed, potentially utilizing the blockchain itself, to gain control of the system temporarily. McCorry et al.\cite{McCorryEtAl} employ smart contracts to facilitate sophisticated bribery attacks, such as transaction censorship and even the modification of the blockchain's history. Other approaches to implementing bribery have also been proposed. For instance, Liao~\cite{Liao2017} shows that bribes can be embedded in transaction fees to incentivize historical rewrites. According to~\cite{SarishtWadhwa2022}, the reverse bribery attack involves miners bribing the parties participating in the exchange to split the profits earned. This results in both parties receiving more collateral than they would have if they had acted honestly. 

HaoChung et al. proposed a solution named PONYTA~\cite{HaoChung2022} to prevent miner collusion with Alice or Bob to gain more collateral. To achieve this, they required Alice and Bob to each deposit $c_a$ and $c_b+v$ into the contract, which is hash locked with $h_a$ and $h_b$ respectively. Here, $c_a > v$ and $c_b> 2v$, where $v$ is the exchange amount. If any party, denoted by $P$, sends ($pre_a$, $pre_b$) such that $H(pre_a) = h_a$ and $H(pre_b) = h_b$, then the solution sends $v$ to party $P$, while all the remaining coins are burnt. In PONYTA, the protocol requires Alice and Bob to deposit collateral exceeding the value of the tokens being exchanged, making it less favorable for users.

Dimitris et al.~\cite{karakostas2024blockchainbribingattacksefficacy} examines bribing attacks in Proof-of-Stake (PoS) distributed ledgers through a game-theoretic lens. However, It does not explicitly model long-term strategic behaviors beyond the immediate bribing rounds.


From this, we conclude that previous works have not explored the scenario where a miner bribes other miners to execute an attack—a key focus of our study. \cite{Ittay2023} analyzed the Bitcoin blockchain and found that selfish miners engaged in bribing lower-status compliant miners to maximize their profits.

\noindent{\bf Fair exchange of collateral over a blockchain}
The authors of~\cite{SarishtWadhwa2022} introduced novel fair exchange protocols that enable reverse bribery. Unlike previous studies that propose fair exchange solely between the two non-mining parties, this approach achieves fair exchange between the asymmetrical parties involved, namely the miner and Bob.

However, The approach presented in~\cite{SarishtWadhwa2022} grants Bob more authority by allowing him to broadcast the completed block, which we leverage in our study to devise a novel attack on MAD-HTLC. Our findings demonstrate that Bob can earn a similar amount of collateral as he would by attacking a naive HTLC.

\noindent{\bf Transaction placement by the miner to create a block: } Miners are usually portrayed as rational actors who prioritize transactions with higher fees to include in their blocks. However, according to ~\cite{SarishtWadhwa2022}, they can achieve better outcomes by actively generating more opportunities, such as by collaborating with certain parties involved in the exchange and sharing the net profit. These miners, who create such opportunities, are referred to as active rational miners. ~\cite{Daian2020Frontrunning} shows one such way to create an opportunity by manipulating transaction orderings and mounting, e.g., frontrunning attacks. This surplus is referred to as {\em Miners Extractable Value (MEV)}. 

In contrast, \prot\ is resistant to such attacks, as the potential rewards for dishonest behavior are lower than those earned by acting honestly.





\begin{appendices}
\section{Discussion \& Future Work}
\label{sec:futurework}
\noindent{\bf Different bribery scenarios:}
Bribery can involve any parties in a protocol, with miners actively seeking bribes or acting rationally based on their transaction pools. Existing models like MAD-HTLC and He-HTLC overlook miner collusion, focusing only on bribery between miners and a single party. Considering these overlooked scenarios, we propose the M2MBA attack, where miners collaborate to share bribery profits, and B3A, a stronger attack on MAD-HTLC allowing Bob to earn rewards comparable to attacking naive HTLC.


More broadly, any application-level protocol that relies on miners as enforcers must account for the risk of active miner participation in bribery and manipulation of application logic. Exploring such miner-driven application-level attacks presents a promising direction for future research.



\noindent{\bf Our solution, i.e. \prot:} Our proposed solution, \prot, mitigates bribery among miners by removing their control over token confiscation and structuring the exchange into two distinct phases. In the first phase, the participating parties commit to either proceed with the exchange or abort it—returning both the exchange tokens and collateral to Bob. In the second phase, this commitment is executed, and the actual token transfer occurs.

While \prot\ addresses bribery scenarios within a single\\ blockchain network, its effectiveness in cross-chain exchanges remains uncertain. Even when following \prot, inter-blockchain exchanges may still be vulnerable to bribery attacks. Adapting \prot\ to securely handle cross-chain exchanges is an open research challenge.

\noindent{\bf On Assumptions:} In our analysis, we make two primary assumptions: (i) only one block is mined per round, and (ii) token values remain constant throughout the protocol execution. The first assumption is reasonable for our attacks, as we only need to demonstrate the existence of scenarios where the attacks succeed. For \prot, since miners no longer have the ability to confiscate tokens, they gain no additional advantage by forking the blockchain. The second assumption holds for exchanges within a single blockchain network. However, when \prot\ is applied to inter-blockchain exchanges, analyzing the impact of token price volatility becomes a crucial research direction, as different tokens may exhibit varying growth rates.


\section{Smart Contract to perform Bribery Attack on naive HTLC} 
In this section, we provide a sample smart contract (Algorithm~\ref{algo:Bribery-Contract}) that demonstrates a bribery attack on a naive HTLC. We will extend this contract to implement the M2MBA attack. The contract is initialized by the attacker (Bob, in our case) using the {\em init} method, with a deposit of $v^{dep}$.

The contract consists of two primary methods: {\em requestBribe} and {\em claimBribe}. The {\em requestBribe} method allows miners to reserve their future bribe by committing to censor $tx_A^{dep}$ before the attack succeeds. Once the attack is successful, {\em claimBribe} facilitates the actual bribe transfer from Bob’s deposit to the participating miners.

Notably, calling {\em claimBribe} requires $pre_A$ as input to prove that Alice has revealed it. If Alice remains silent until time $T$, miners cannot claim the bribe, and Bob is refunded his full deposit of $v^{dep}$.



\begin{algorithm}[h!]
  \caption{Bribery Contract ($C_{Bob}$)}\label{algo:Bribery-Contract}
  \begin{algorithmic}[1]
    \State {\bf Variables: }
    \State $v^{dep} \gets 0$; $balLeft \gets 0$;
    
    \State $Bal \gets \{\}$ {\color {gray}//mapping from address to balances}\\
    
    \State {\bf Constants: }
    \State $br = f^{dep}_A + \epsilon$\\

    \State {\bf Functions: }
    \State {\color {gray}//Called during contract deployment}
    \Procedure{init}{ $val$ }
        \State $v^{dep} \gets val$; $balLeft \gets v^{dep}$
    \EndProcedure\\

    \State {\color {gray}//Called by the miners to lock their bribe}
    \Procedure{requestBribe}{ }
        \If{caller is not the miner of the current block {\bf or} function is called more than once in the block}
        \State {\bf return}
        \EndIf

        \If{height of the current block is $\leq T$ }
        \State $Bal[caller] \gets Bal[caller] +1$
        \State $balLeft \gets balLeft - br$
        \EndIf
    \EndProcedure\\

    \State {\color {gray}//Anyone can call to claim their bribe}
    \Procedure{claimBribe}{$preImage$}
        \State $count \gets 0$
        
        \If{$preImage$ is not a $pre_A$ {\bf or} $tx^{dep}_A$ is included in the block {\bf or} $tx^{dep}_B$ is not included in the block till the current block {\bf or} $balLeft-(br+f_B^{dep}+f_B^{C_{Bob}})<0$} 
            \State {\bf return}
        \EndIf

        \For{Each $(addr, bal) \in Bal$}
            \State Send $br*bal$ token of bribe  to $addr$
            \State $count \gets count +1$

        \EndFor
        \State {\color {gray}//Gain to the caller}
        \State send $br$ to the caller (miner).
        \State {\color {gray}//Send the remaining tokens to Bob}
        \State send $v_{dep} - (count*br +br)$ to Bob 
         
        \State $Bal \gets \{\}$ {\color {gray}//Reset the values}
        \State $v^{dep} \gets 0$
    \EndProcedure
    
    \State {\color {gray}//Refund to the bob if miner include $tx^{dep}_A$}
    \Procedure{refundToBob}{ $val$ }
        \If{$tx^{dep}_A \in$ block with block height $\leq T$}
        \State send $v^{dep}$ to Bob.
        \EndIf
    \EndProcedure\\
  \end{algorithmic}
\end{algorithm}

\section{The success-dependent reverse bribery attack (SDRBA) and  Hybrid delay-reverse bribery attack
(HyDRA).}
\label{apx:SDRBA}

\begin{figure*}[t!]
    \centering
    \begin{subfigure}{0.29\linewidth}
    \centering
    \includegraphics[height=3.25cm, width=0.98\linewidth]{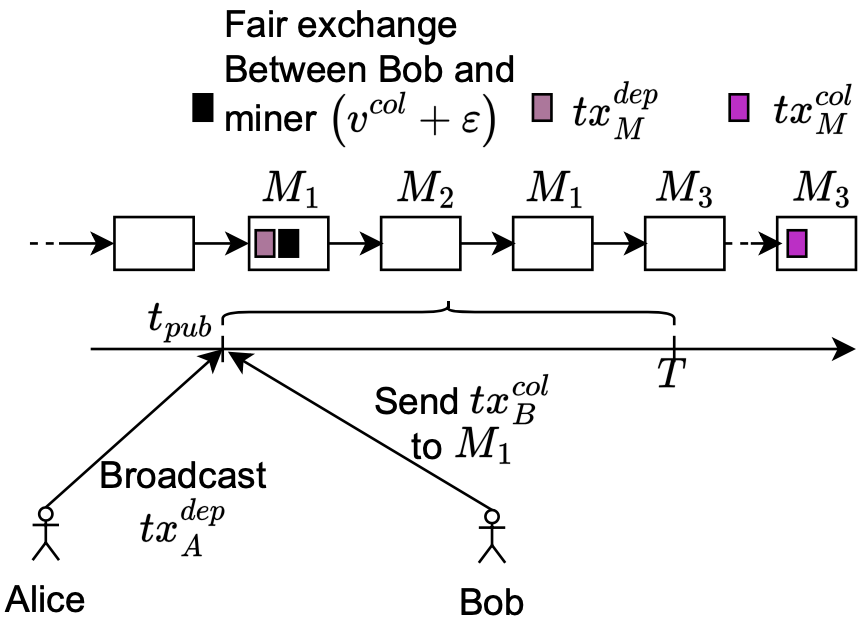}
    \caption{}
    \label{fig:MD_HTLC_SDRBA}
    \end{subfigure}
    \begin{subfigure}{0.28\linewidth}
    \centering
    \includegraphics[height=3.25cm, width=0.98\linewidth]{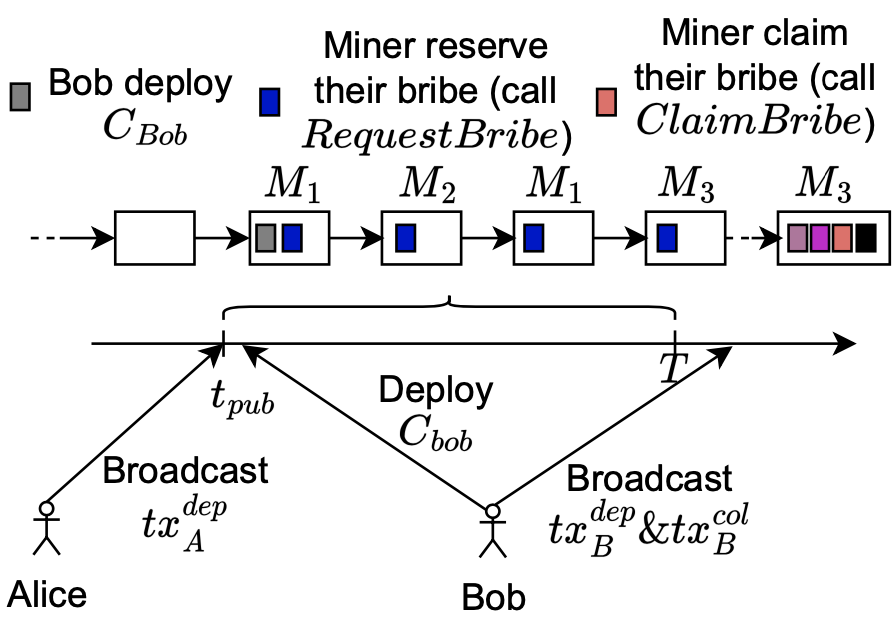}
    \caption{}
    \label{fig:MD_HTLC_HyDRA}
    \end{subfigure}
    \begin{subfigure}{0.34\linewidth}
    \centering
    \includegraphics[height=3.25cm, width=0.98\linewidth]{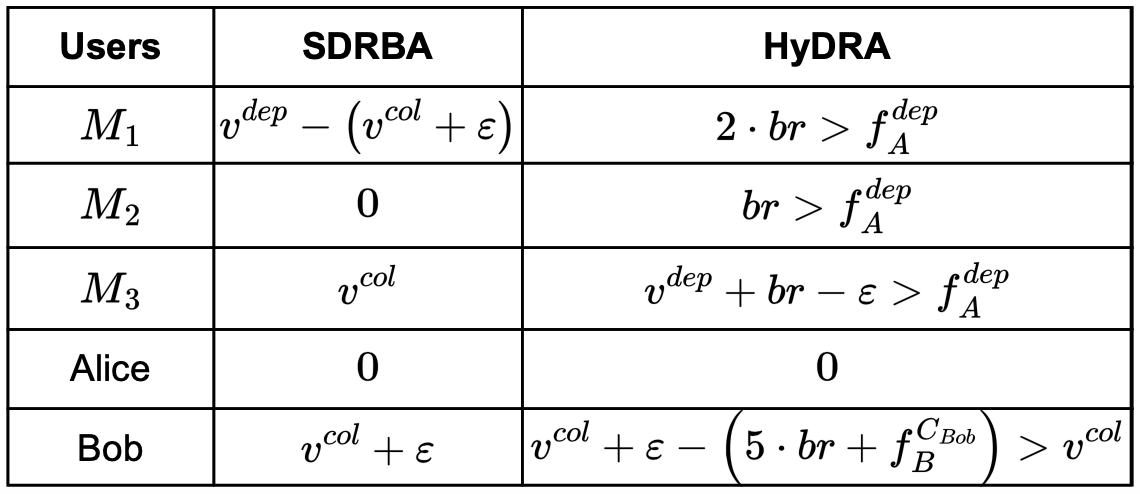}
    \caption{}
    \label{fig:MD_HTLC_SDRBA_HyDRA_table}
    \end{subfigure}
    \caption{(a)SDRBA on MAD-HTLC, (b)HyDRA on MAD-HTLC, (c)Utility gained by each party after performing SDRBA and HyDRA on MAD-HTLC}
    \label{fig:MD_HTLC_SDRBA_HyDRA}
\end{figure*}

The success-dependent reverse bribery attack (SDRBA) demonstrates that a reverse bribery attack can be successful when $v^{dep} > v^{col}$. In SDRBA, a miner, $M_1$ in our example in Figure~\ref{fig:MD_HTLC_SDRBA}, pays a bribe ($>v^{col}$) to Bob only if it can redeem $v^{dep}$ through {\em dep-M}. However, $M_1$ must reveal the $pre_A$ and $pre_B$ while redeeming $v^{dep}$, so all miners have access to $pre_A$ and $pre_B$, and they compete to redeem $v^{col}$ after timeout $T$. The probability of miner $M_1$ mining the block and redeeming $v^{col}$ is $\lambda_{M_1}/\lambda$, where $\lambda_{M_1}$ is $M_1$'s block generation rate. Therefore, in the worst case of SDRBA, where another miner redeems $v^{col}$ through {\em col-M}, the bribing miner ($M_1$) earns $v^{dep} -(v^{col} + \epsilon)$, as depicted  in Figure~\ref{fig:MD_HTLC_SDRBA_HyDRA_table}. Consequently, if $v^{dep} \leq v^{col}$, the attack may not be advantageous.

To overcome the disadvantage of SDRBA, the HyDRA is proposed which is designed to work for any values of $v^{dep}$ and $v^{col}$. In HyDRA, Bob initiates the process by bribing miners to exclude transaction $tx^{dep}_A$ until time $T$, at a cost of $k*br$, where $k$ is the number of blocks mined in time $T-t_{pub}$. After time $T$, a miner (For example, $M_3$ in the Figure~\ref{fig:MD_HTLC_HyDRA}) engages in SDRBA, where it exchanges $v^{col} + \epsilon$ with Bob and redeem $v^{dep} + v^{col}$.
So Bob end up earning $v^{col} + \epsilon - [(k+1)*br+f^{C_{Bob}}_B]$, where $\epsilon > [(k+1)*br+f^{C_{Bob}}_B]$. The example depicted in the Figure~\ref{fig:MD_HTLC_SDRBA_HyDRA} has $k = 4$, so it requires $\epsilon > 5*br$. Alice can also collude with a miner to perform HyDRA by deploying the contract, say $C_{Alice}$, and earn $v^{dep} + \epsilon- [(k+1)*br+f^{C_{Alice}}_A]$ in a similar manner. 

\section{Solo mining vs Pool mining}
\label{apx:soloVsPool}


This section explains why miners prefer pool mining over solo mining. The key factor is reward variance—pool mining offers lower variance, leading to more stable and predictable income. In probabilistic systems like cryptocurrency mining, this stability is often valued over higher but inconsistent solo rewards. We begin by outlining the underlying assumptions.\\

\noindent{\bf Assumptions and Notations: }
Let \( E_{\text{solo}} \) and \( E_{\text{pool}} \) denote the expected rewards from solo and pool mining, respectively. Let \( H \) be the total network hash rate, \( h \) the miner's hash rate, \( N \) the number of miners in the pool, \( R \) the fixed block reward, and \( f \) the pool fee. Let \( \lambda \) and \( \lambda_i \) be the average mining rates of the network and miner \( i \), respectively. Given these, we analyze a miner’s rewards under solo and pool mining.\\

\noindent{\bf Rewards earned by the miner in solo mining and pool mining: }
The individual miner’s probability of successfully mining a block $X$ in a given time period follows: \[
X \sim \text{Poisson}(\lambda_i), \quad \text{where } \lambda_i = \frac{h}{H} \cdot \lambda.
\]
So the expected reward for the miner in solo mining using the expected value of a Poisson distribution ($E[X] = \lambda_i$) is $E_{\text{solo}} = \lambda_i \cdot R = \frac{h}{H} \cdot \lambda \cdot R.$

Similarly, in pool mining, miners combine their hash rates to form a larger effective hash rate. The pool’s combined hash rate is $N \cdot h$, and the probability of the pool mining $k$ blocks follows:
\[
\lambda_{\text{pool}} \sim \text{Poisson}(\lambda_{\text{pool}}), \quad \text{where } \lambda_{\text{pool}} = \frac{N \cdot h}{H} \cdot \lambda.
\]

Each miner in the pool receives a share of the reward proportional to their contribution $h$. The expected reward for an individual miner in the pool, after accounting for the pool fee $f$, is:
\[
E_{\text{pool}} = (1 - f) \cdot \left( \frac{h}{N \cdot h} \cdot \lambda_{\text{pool}} \cdot R \right)=(1 - f) \cdot \frac{h}{H} \cdot \lambda \cdot R.
\]

Given this we are all set to find a relation between $E_{\text{pool}}$ and $E_{\text{solo}}$. To compare $E_{\text{solo}}$ and $E_{\text{pool}}$, we compute the ratio:
\[
\frac{E_{\text{solo}}}{E_{\text{pool}}} = \frac{\frac{h}{H} \cdot \lambda \cdot R}{(1 - f) \cdot \frac{h}{H} \cdot \lambda \cdot R} = \frac{E_{\text{solo}}}{E_{\text{pool}}} = \frac{1}{1 - f}.
\]

Since $f > 0$, we have $\frac{1}{1 - f} > 1$, implying $E_{\text{solo}} > E_{\text{pool}}$. Thus, the expected reward for solo mining is higher in absolute terms, but this does not account for variance. Most miners are risk-averse, meaning they prefer smaller, steady rewards over infrequent, larger payouts. A risk-averse miner values stability and is willing to sacrifice a small portion of their expected reward (pool fee) for predictable income. This is because of the higher variance in the rewards in case of solo mining. So next we will consider the variance. 


For the further discussion, let $\text{Var}_{\text{solo}}$ and $\text{Var}_{\text{pool}}$ be the variance in case of solo mining and pool mining, respectively. So its easy to conclude, $\text{Var}_{\text{pool}}= \frac{\text{Var}_{\text{solo}}}{N}$, where $N$ is the number of miners in the pool. 

As we have already discussed, a risk-averse miner values stability, represented by a utility function $U(W)$, where $W$ is the miner's reward. A common choice is the negative exponential utility function,  $U(W) = -e^{-\alpha W}$, where $\alpha > 0$ represents the miner’s degree of risk aversion.

For small variances, the expected utility can be approximated using a Taylor expansion: $E[U(W)] \approx U(E[W]) + \frac{1}{2} U''(E[W]) \cdot \text{Var}(W)$,
where $U''(E[W]) = -\alpha^2 e^{-\alpha E[W]}$.
Thus,
\[
E[U(W)] \approx -e^{-\alpha E[W]} \left( 1 - \frac{\alpha^2 \text{Var}(W)}{2} \right).
\]

Given this, we will next calculate the utility for a solo mining and a pool mining

\[
E[U(W_{\text{solo}})] \approx -e^{-\alpha E_{\text{solo}}} \left(1 - \frac{\alpha^2 \text{Var}_{\text{solo}}}{2}\right).
\]

$E[U(W_{\text{pool}})] \approx -e^{-\alpha E_{\text{pool}}} \left(1 - \frac{\alpha^2 \text{Var}_{\text{pool}}}{2}\right) = -e^{-\alpha (1 - f)E_{\text{solo}}} \left(1 - \frac{\alpha^2 \text{Var}_{\text{solo}}}{2N}\right)$

Now we have utility functions for solo mining and pool mining, i.e. $E[U(W_{\text{solo}})]$ and $E[U(W_{\text{solo}})]$. Next we will compare these utility functions to reach to the conclusion. We compare the utilities by calculating the difference,

$\Delta U = E[U(W_{\text{pool}})] - E[U(W_{\text{solo}})]$

\[ 
 = -e^{-\alpha (1 - f)E_{\text{solo}}} \left(1 - \frac{\alpha^2 \text{Var}_{\text{solo}}}{2N}\right) - -e^{-\alpha E_{\text{solo}}} \left(1 - \frac{\alpha^2 \text{Var}_{\text{solo}}}{2}\right)
\]

In the above equation, expanding $e^{-\alpha (1 - f)E_{\text{solo}}}$ around $E_{\text{solo}}$ for small $f$, i.e, $e^{-\alpha (1 - f)E_{\text{solo}}} \approx e^{-\alpha E_{\text{solo}}} (1 + \alpha f E_{\text{solo}})$, will gives the approximate value of $\Delta U$ after simplification. 

\[
\Delta U \approx -e^{-\alpha E_{\text{solo}}} \left[ \alpha f E_{\text{solo}} - \frac{\alpha^2 \text{Var}_{\text{solo}}}{2} + \frac{\alpha^2 \text{Var}_{\text{solo}}}{2N} \right].
\]

We can interpret the above equation as follows:

\begin{itemize}
    \item \textbf{First Term} ($\alpha f E_{\text{solo}}$): Represents the utility reduction due to the pool fee $f$. This term reduces the miner’s expected utility in pool mining.
    \item \textbf{Second Term} ($-\frac{\alpha^2 \text{Var}_{\text{solo}}}{2}$): Represents the penalty for high variance in solo mining.
    \item \textbf{Third Term} ($\frac{\alpha^2 \text{Var}_{\text{solo}}}{2N}$): Represents the reduced penalty due to variance in pool mining. For large $N$, this term is small, making pool mining more stable.
\end{itemize}

Given this, we can conclude that, If $f$ is small, the utility reduction due to the pool fee is negligible. Furthermore, For large $N$, the variance penalty in pool mining becomes significantly smaller than in solo mining ($\frac{1}{N}$-factor reduction). Lastly, Risk-averse miners prioritize stability, and the reduced variance in pool mining ensures that $\Delta U > 0$ for reasonable values of $f$, $N$, and $\text{Var}_{\text{solo}}$. Thus, over time, the stability from lower variance outweighs the small reduction in expected rewards due to the pool fee, making pool mining more lucrative for risk-averse miners.

\section{M2MBA Analysis Lemma's proofs}
\label{sec:M2MBAproof}

\noindent{\begin*{\bf Lemma 1}
\label{lm:1}
In $\zeta^{mba}(k, red)$, where $k \leq T$, if $(T-t_{pub})*br*\lambda_i/\lambda_{col}>f^{dep}_A$, then it is a dominant move for active miners to accept the bribe by censoring the transaction $tx^{dep}_A$ instead of including it in the block.
\end*}

\begin{proof}
    By including $tx_A^{dep}$ miner would earn $f^{dep}_A$. On the other hands, the expected bribe that a miner would earn by censoring the transaction for a block is $br*\lambda_i/\lambda_{col}$.
    $T-t_{pub}$ rounds produce $T-t_{pub}$ blocks. So the total expected bribe that a miner earn is  $(T-t_{pub})*br*\lambda_i/\lambda_{col}$. According to our assumption this value exceeds $f^{dep}_A$. So accepting the bribe is the dominant move.
\end{proof}

\noindent{\begin*{\bf Lemma 2}
\label{lm:2}
In $\zeta^{mba}(k, red)$, where $k \leq T$, if $v^{col}*\lambda_i/\lambda{col} + (T-t_{pub})*br*(2*\lambda_i/\lambda_{col} - 1)>f^{dep}_A$, then it is a dominant move for active miners to offer a bribe to other miners to censor the transaction $tx^{dep}_A$ and perform the attack rather than including it in the block.
\end*}

\begin{proof}
    The attack consists of two phases. In the first phase, a bribe is offered to censor the transaction $tx^{dep}_A$. In the second phase, the miner attempts to mine the block that includes the transaction $tx^{col}_M$ in round $T+1$, where Bob releases $tx^{dep}_B$ by revealing $pre_B$. The miner, referred to as the confiscating miner $M_i$, will earn a expected reward of $v^{col}*\lambda_i/\lambda{col}$ by successfully mining the block that includes $tx^{col}_M$.

    The confiscating miner $M_i$ offers a bribe to other miners, including himself, to censor the transaction $tx^{dep}_A$ for round $T-t_{pub}$. 
    The bribe that $M_i$ offers to himself is $(T-t_{pub})*br*\lambda_i/\lambda{col}$. On the other hands, the bribe that $M_i$ offers to other active miners is $(T-t_{pub})*br*(1- \lambda_i/\lambda{col})$.
    
    So the net gain of the miner by offering a bribe and perform the attack is $v^{col}*\lambda_i/\lambda{col} + (T-t_{pub})*br*\lambda_i/\lambda{col} - (T-t_{pub})*br*(1- \lambda_i/\lambda{col}$), which boils down to $v^{col}*\lambda_i/\lambda{col} + (T-t_{pub})*br*(2*\lambda_i/\lambda_{col} - 1)$. According to our assumption this value exceeds $f^{dep}_A$, which is the incentive that miner earns by including $tx_A^{dep}$ in the block. So offering a bribe and performing the attack is the dominant move.




\end{proof}

\noindent{\begin*{\bf Lemma 3}
\label{lm:3}
        In $\zeta^{mba}(k, red)$, where $k \leq T$. If $v_{col}*\lambda_i > f^{dep}_A$, it is the dominant move for passive miners to censor $tx^{dep}_A$ rather than include it in the block.
\end*}
\begin{proof}
    As passive miner will silently wait for the Bob to release $tx^{dep}_B$ that reveals $pre_B$. Passive miners are not involved in the either accepting bribe or offering a bribe. So the net gain of the passive miner, $M_i$, is $\lambda_i*v^{col}$. According to our assumption this value exceeds $f^{dep}_A$, which is the incentive that miner earns by including $tx_A^{dep}$ in the block. So censoring $tx_A^{dep}$ and performing the attack is the dominant move for a passive miner.
\end{proof}

\noindent{\begin*{\bf Lemma 4}
\label{lm:4}
In $\zeta^{mba}(k, red)$, where $k > T$. Let $z$ be the net earning of the miner by confiscating $v^{col}$. If $z > f^{dep}_A$, and Bob release transaction $tx^{dep}_B$ and hence release secret $pre_B$ at $T+1$, it is the dominant move for all miners to mine the block that includes $tx^{col}_M$, at $k=T+1$. (transient transition from $\zeta^{mba}(k, red)$ to $\zeta^{mba}(k, nred$-$rev$)  by including $tx^{dep}_B$ and then including $tx^{col}_M$ in the same block)
\end*}

\begin{proof}
    If miner $M_i$ decide to defer the inclusion of $tx^{col}_M$ in the block till round $t_y$ in future. $M_i$ must  ensure either of the following two things:

        \begin{itemize}
            \item $M_i$ must mine all the block from round $T+1$ to $t_y$.
            
            \item If $M_i$ couldn't mine all the block between round $T+1$ to $t_y$, they should insure that other miner who mined the block shouldn't include $tx^{col}_M$, which they can ensure by bribing them more than $z$.
        \end{itemize}

        In the former scenario, the probability of block mined by the miner $M_i$ decreases exponentially with increase of difference between $t_y$ and $T+1$ and hence the expected gain would also decreases exponentially, i.e. $z*x^{[t_y - (T+1)]}$, where x is the miner's mining power, which is $\lambda_i$ for passive miner and $\lambda_i/\lambda_{col}$ for active miners.
        

        On the other hands, in the later scenario, $M_i$ has to bribe more than $z$ to the other miner for not including $tx^{col}_M$. The expected bribe that $M_i$ needs to pay is $z*(1-x^{[t_y - (T+1)]})$, which will also increase with increase of difference between $t_y$ and $T+1$. Note that the second condition, where miners offers a bribe greater than $z$, is specific to the active miners.

        Combining the above two conditions, we can claim that it is dominant move for all the miner to mine the block that includes $tx^{col}_M$ at $T+1$, if $tx^{dep}_B$ is released, i.e., $pre_B$ is revealed. 
        

        
\end{proof}


\noindent{\begin*{\bf Lemma 5}
\label{lm:5}
We made the valid assumption that for an active miner, inequality $v^{col}\lambda_i/\lambda_{col} + (T - t_{pub})*br*(2\lambda_i/\lambda_{col} - 1) > f^{dep}_A$ holds. Similarly, for a passive miner, Assumption $v^{col}*\lambda_i > f^{dep}_A$ is also valid. 
\end*}
\begin{proof} We will prove this using a combination of theoretical analysis and empirical evidence.
   We will divide the proof into two parts: the first will focus on the active miners, and the second on the passive ones.  

   \begin{itemize}
       \item $Case_1-Active \text{ } miners:$ As we discussed earlier the active miners receives a bribe from the bribing miner, i.e. the miner who confiscate $v^{col}$, which offers a bribe to other active miners.
       We know that the value $v^{col}\lambda_i/\lambda_{col} + (T - t_{pub})*br*(2\lambda_i/\lambda_{col} - 1)$ indicates the amount the bribing miner gets after paying bribe to other active miners, including himself. So to prove this value is greater than $f^{col}_A$, its sufficient to prove that the bribe that every active miner gets from the bribing miner is greater than $f^{col}_A$ and the value left with bribing miner after paying bribe to other active miners, including himself, is greater than 0. This we are going to prove next.

       If the game is in $\zeta^{mba}(k, red)$ state, the expected number of blocks that the active miner, say $M_i$, with mining power $\lambda_i$, mines among $T-t_{pub}$ blocks is $(T-t_{pub})*\lambda_i/\lambda_{col}$.

       This scenario can be viewed as a series of $T - t_{pub})$ Bernoulli trials with success probability $\lambda_i/\lambda_{col}$. The number of successes is therefore Binomially distributed, and the expected number of blocks miner $M_i$ creates is $(T-t_{pub})*\lambda_i/\lambda_{col}$.

       According to the contract in Algorithm~\ref{algo:M2MBA} for each block miner would receive the bribe of $f^{dep}_A *\frac{\lambda_{col}}{\lambda_i}*\frac {1}{T-t_{pub}} + \epsilon$ for censoring $tx^{dep}_A$. So the net bribe that an active miner with mining power $\lambda_i$ earn is $(\frac{\lambda_{i}}{\lambda_{col}}*(T-t_{pub}))*(f^{dep}_A *\frac{\lambda_{col}}{\lambda_i}*\frac {1}{T-t_{pub}} + \epsilon)$. This equation boils down to  $f^{dep}_A+(\frac{\lambda_{i}}{\lambda_{col}}*(T-t_{pub}))*\epsilon$. So for any $\epsilon>0$ miner would earn more than $f^{dep}_A$.

       Furthermore, the contract in the algorithm\ref{algo:M2MBA} (line 26) ensures that the amount left with the bribing miner after paying bribe to all the active miners, {\em including} himself is positive. Hence the resulting gain for the miner would always be greater than $f^{dep}_A$. 
       
       

    \item $Case_2-Passive \text{ } miners: $ For passive miners, we'll support our claim with empirical evidence. The Ethereum fee is less than 0.01\% of the transaction value \cite{EthereumTxFee}. Additionally, the miner with the lowest mining power among the top 50 miners from the production Ethereum network—who collectively contribute 99.98\% of the total mining power—has a mining power of 0.015 \cite{Renoir}. Even the miner with the lowest mining power would have expected earning more than the transaction fee by censoring $tx^{dep}_A$, so our claim that $v^{col}*\lambda_i > f^{dep}_A$ holds, even for the miner with the smallest share of mining power.

   \end{itemize}
\end{proof}

\section{\prot\ Analysis Lemma's proofs}
\label{sec:DEMBAproof}

\noindent{\begin*{\bf Lemma 6}
\label{lm:AliceDominant}
For $v^{ded}>0$, it is dominant move for the Alice to follow \prot\ than deviating from it. 

\end*}


\begin{proof}
    The possible moves for Alice to redeem $C^{col}_A$ are as follows: 1) Release $pre_A$ before round $T$, where she would  earn $v^{dep}$ and $v^{col}_A$ (potential transition to $nred\text{-}AB$ or $nred\text{-}ABT$, based on when Bob reveal $pre_B$), 2) Release $pre'_A$ after round $T$, where she would earn $v^{col}_A$, but ends up paying more transaction fee (from eq.~\ref{eq:fee}) and 3) Release both $pre'_A$ and $pre_A$ after round $T$ where she would earn $v^{col}_A-v^{ded}$ (potential transition to $nred\text{-}AA'B$ or $nred\text{-}AA'BT$. 
    
    The second scenario occurs when Alice is genuinely offline and misses the deadline $T$. In this case, \prot\ still imposes a penalty by requiring her to pay a higher fee for the transaction that reveals $pre_A'$ after round $T$. The third scenario involves potential malicious behavior by Alice, where she pretends to have already released $pre_A$ despite not having done so, intending to deliberately burn $v^{dep}$. However, it becomes evident that this approach results in Alice earning an $v^{ded}$ less than if she had acted honestly.

    Based on the above discussion, we conclude that in the states $nred\text{-}AB$ or $nred\text{-}ABT$, Alice's earnings are higher compared to other states in the game.


    
\end{proof}

\noindent{\begin*{\bf Lemma 7}
\label{lm:BobDominant}
For $v^{ded}>0$, it is dominant move for the Bob to follow \prot\ than deviating from it.

\end*}

\begin{proof}
    The possible moves for Bob to redeem $C^{col}_B$ is to release $pre_B$. However, Bob can do so anytime. So possible malicious behavior from Bob is to deliberately delay this release. We know that delaying it beyond round $T$ would make him earn $v^{col}_B-v^{ded}$ which is less than $v^{col}_B$, i.e. the value he would have earned by releasing it before $T$.

    Based on the above discussion, we conclude that in the states $nred\text{-}AB$, $nred\text{-}A'B$ and $nred\text{-}AA'B$, Bob's earnings are higher compared to other states in the game.
\end{proof}

\noindent{\begin*{\bf Lemma 8}
\label{lm:MinerDominant}
For $\alpha<1$, it is dominant move for the miners to follow \prot\ than deviating from it.

\end*}
\begin{proof}
    A potential malicious behavior by the miner is to exclude Alice's and Bob's transactions, which reveal $pre_A$ and $pre_B$, respectively, from the block—either indefinitely or until round $T$. However, since \prot\ transactions offer higher fees than unrelated transactions (by assumption), the case of indefinite exclusion is unlikely. Moreover, from Eq.~\ref{eq:minerFee}, it is evident that the miner earns less by including \prot\ transactions after round $T$.


    
    We can quantify this further as, for $t > T$, the miner earns $\alpha^t \cdot f^{col}_{pre_A}$ by including $tx^{col}_{pre_A}$ and $\alpha^t \cdot f^{col}_{pre_B}$ for $tx^{col}_{pre_B}$. Thus, timely inclusion of transactions ensures greater profitability for the miner.

    Based on the above discussion, we conclude that in the states $nred\text{-}AB$, $nred\text{-}A'B$ and $nred\text{-}AA'B$, miner's earnings are higher compared to other states in the game.
\end{proof}

\section{M2MBA Smart Contract}
\label{sec:M2MBASmartContract}

This section describes a smart contract designed for executing M2MBA, where the initialization method in Algorithm~\ref{algo:Bribery-Contract} is replaced with {\em lockCollateral}. Each participating miner locks their collateral, which is refunded either upon the successful completion of the attack or when {\em refundToMiners} is invoked after the timeout $T$ (as specified in line 37 of Algorithm~\ref{algo:M2MBA}).

Similar to Algorithm~\ref{algo:Bribery-Contract}, the {\em requestBribe} method allows miners to reserve their future bribe by committing to censor $tx_A^{dep}$ before the attack succeeds. Upon successful completion of the attack, the {\em claimBribe} method enables the actual bribe transfer from Bob’s deposit to the participating miners.

\begin{algorithm}[h!]
  \caption{M2MBA Contract ($C_{M2M}$)}\label{algo:M2MBA}
  \begin{algorithmic}[1]
    \State {\bf Variables: }
    \State He-Col: Address of He-col contracts in He-HTLC
    \State $v^{col} \gets 0$, $Bal_A \gets \{\}$, $count \gets 0$.
    \State $LockCol \gets \{\}$ {\color {gray}//maps from balance to address}\\
    
    \State {\bf Constants: }
    \State $br = f^{dep}_A * \frac{\lambda_{col}}{\lambda_i}*\frac {1}{T-t_{pub}} + \epsilon$\\

    \State {\bf Functions: }
    \State {\color {gray}//Called while contract deployment and amount ($v^{col}$) locking}

    \Procedure{lockCollateral}{ $val$ }
        \State $LockCol[Caler]  \gets val$
    \EndProcedure\\
    
    \State {\color {gray}//Called by the miners to lock their bribe}
    \Procedure{requestBribe}{ }
        \If{Caller is not the miner of the current block {\bf or} Function is called more than once in the block}
        \State {\bf return}
        \EndIf

        \State $Bal_A[caller] \gets Bal_A[caller] +1$, $count \gets count +1$ 
    \EndProcedure\\

    \State {\color {gray}//Anyone can call to claim their bribe}
    \Procedure{claimBribe}{$preImage_A$}
      
        \State $Addr_M$ address of miner who redeem He-Col
        
        \If{$H(preImage_A)$ is not a $H(pre_A)$ {\bf or}  $tx^{dep}_A$ is included in the block $\leq T$ {\bf or} He-Col is not redeemed till the current block through Col-M {\bf or} $v^{col}-count*br+Bal_A[Addr_M]*br < 0$}
            \State {\bf return}
        \EndIf
    
        

        \State{\color {gray}//Bribe for censoring $tx^{dep}_A$}
        \For{Each $(addr, bal) \in Bal_A$}
            \State Send $br*bal$ token of bribe  to $addr$
            \State $LockCol[Addr_M]-=br*bal$

        \EndFor


        \State{\color {gray}//Incentive to the caller for calling ClaimBribe}
        \State send $br>f$ to caller
        , $LockCol[Addr_M]-=br$
        
        \State{\color {gray}//Send remaining and locked collateral to miners}
        \For{Each $(addr, bal) \in LockCol$}
            \State Send $bal$ token  to $addr$
        \EndFor
        \State $LockCol \gets \{\}$, $Bal \gets \{\}$, $v^{col} \gets 0$ {\color {gray}//Reset values}
        
    \EndProcedure\\

    \State {\color {gray}//Refund locked collaterals to the miners}
    \Procedure{refundToMiners}{ $val$ }
        \If{$tx^{dep}_A \in$ block with block height $\leq T$}
            \For{Each $(addr, bal) \in LockCol$}
            \State Send $bal$ token to $addr$
        \EndFor
        \EndIf
    \EndProcedure\\
  \end{algorithmic}
\end{algorithm}

\end{appendices}

\end{document}